\documentclass[pdflatex,sn-basic]{sn-jnl}

\usepackage{graphicx}%
\usepackage{multirow}%
\usepackage{amsmath,amssymb,amsfonts}%
\usepackage{amsthm}%
\usepackage{mathrsfs}%
\usepackage[title]{appendix}%
\usepackage{xcolor}%
\usepackage{textcomp}%
\usepackage{manyfoot}%
\usepackage{booktabs}%
\usepackage{algorithm}%
\usepackage{algorithmicx}%
\usepackage{algpseudocode}%
\usepackage{listings}%
\usepackage{float}

\theoremstyle{thmstyleone}%
\newtheorem{theorem}{Theorem}
\newtheorem{lem}{Lemma}

\newtheorem{prp}{Proposition}
\newtheorem{asm}{Assumption}
\theoremstyle{thmstyletwo}%

\theoremstyle{thmstylethree}%

\theoremstyle{remark}

\newcommand{\Var}{{\rm Var}}

\newcommand{\E}{\mathbb{E}}
\newcommand{\R}{\mathbb{R}}

\newcommand{\argmin}{\mathop{\rm argmin}\limits}

\newcommand{\ip}[2]{\left\langle #1,\,#2 \right\rangle}
\newcommand{\norm}[1]{\left\lVert #1 \right\rVert}
\newcommand{\coloneqq}{\mathrel{\mathop:}=}
\newcommand{\ind}{\mathbf{1}}

\begin{document}

\title[Sampling from density power divergence-based generalized posterior distribution via stochastic optimization]{Sampling from density power divergence-based generalized posterior distribution via stochastic optimization}


\author*[1]{\fnm{Naruki} \sur{Sonobe}}\email{narukisonobe@gmail.com}

\author[1]{\fnm{Tomotaka} \sur{Momozaki}}\email{t\_momozaki@rs.tus.ac.jp}
\equalcont{These authors contributed equally to this work.}

\author[2,3]{\fnm{Tomoyuki} \sur{Nakagawa}}\email{tomoyuki.nakagawa@meisei-u.ac.jp}
\equalcont{These authors contributed equally to this work.}

\affil*[1]{\orgdiv{Department of Information Sciences}, \orgname{Tokyo University of Science}, \orgaddress{\street{2641 Yamazaki}, \city{Noda-shi}, \postcode{278-8510}, \state{Chiba}, \country{Japan}}}

\affil[2]{\orgdiv{School of Data Science}, \orgname{Meisei University}, \orgaddress{\street{2-1-1 Hodokubo}, \city{Hino}, \postcode{191-8506}, \state{Tokyo}, \country{Japan}}}

\affil[3]{\orgdiv{Statistical Mathematics Unit}, \orgname{RIKEN Center for Brain Science}, \orgaddress{\street{2-1 Hirosawa}, \city{Wako City}, \postcode{351-0198}, \state{Saitama}, \country{Japan}}}


\abstract{Robust Bayesian inference using density power divergence (DPD) has emerged as a promising approach for handling outliers in statistical estimation. Although the DPD-based posterior offers theoretical guarantees of robustness, its practical implementation faces significant computational challenges, particularly for general parametric models with intractable integral terms. These challenges are specifically pronounced in high-dimensional settings, where traditional numerical integration methods are inadequate and computationally expensive. Herein, we propose a novel {approximate} sampling methodology that addresses these limitations by integrating the loss-likelihood bootstrap with a stochastic gradient descent algorithm specifically designed for DPD-based estimation. Our approach enables efficient and scalable sampling from DPD-based posteriors for a broad class of parametric models, including those with intractable integrals. We further extend it to accommodate generalized linear models. Through comprehensive simulation studies, we demonstrate that our method efficiently samples from DPD-based posteriors, offering superior computational scalability compared to conventional methods, specifically in high-dimensional settings. The results also highlight its ability to handle complex parametric models with intractable integral terms.}




\maketitle

\section{Introduction}

Bayesian inference is recognized as a powerful statistical methodology that is used in various fields. However, its standard posterior distribution suffers from a major weakness: it is highly sensitive to outliers, which can significantly distort estimation results. To address this limitation, researchers have proposed robust alternatives to the standard posterior distribution. A key development was introduced by \cite{bissiri2016general}, who established a general framework for updating beliefs using general loss functions instead of likelihoods. { In addition, it is well known that proper scoring rules yield estimators with built-in robustness to outliers \citep{dawid2016score}. A Bayesian formulation based on scoring rules has been developed by \citet{giummole2019score}.}

This framework has led to recent advancements in generalized posterior distributions that incorporate robust loss functions for effective outlier handling. {\citet{hooker2014robust} introduced Bayesian methods based on robust disparities, launching divergence-based robust inference. Building on this line,} \cite{ghosh2016robust} developed a generalized posterior distribution based on the density power divergence (DPD, also known as $\beta$-divergence {or Tsallis Score}) \citep{basu1998robust, eguchi2001robustifing, mihoko2002robust, dawid2016score}, demonstrating its theoretical robustness. \cite{jewson2018principles} further advanced the field by establishing a comprehensive framework for Bayesian inference using general divergence criteria, showing how to achieve robust inference while maintaining Bayesian principles. 
Additionally, \cite{nakagawa2020robust} introduced a robust Bayesian estimation method using $\gamma$-divergence \citep{fujisawa2008robust}, which exhibited superior performance in variance parameter estimation compared to DPD-based methods. Recently, \cite{matsubara2022robust} proposed a robust approach employing a Stein discrepancy, effectively addressing both model misspecification and intractable normalizing constants in likelihood functions.

Within the robust Bayesian inference framework, the DPD-based posterior, introduced by \cite{ghosh2016robust}, has emerged as a particularly influential approach {\citep{jackjewson2024StabilityGeneralBayesian}}. Although alternative methods such as $\gamma$-divergence may offer stronger theoretical properties, the DPD-based posterior's theoretical guarantees for robustness against outliers and { its predictive stability to model's specification} make it a natural first choice when considering robust estimation methods. Notably, its robustness properties can be flexibly controlled through a single tuning parameter, enabling practitioners to optimize the trade-off between efficiency and robustness. The selection of this crucial tuning parameter has been extensively studied \citep{warwick2005tuning, basak2021tuning, yonekura2023adaptation}, thereby complementing the theoretical foundations of the method. 

Various sampling approaches, including the Metropolis-Hastings (MH) algorithms \citep{ghosh2022robust}, importance-sampling techniques \citep{ghosh2016robust, nakagawa2020robust}, and variational Bayesian inference methods \citep{knoblauch2018beta}, have been developed for the DPD-based posterior. These methodological advances have enabled successful applications across various domains, including online changepoint detection \citep{knoblauch2018beta} and differential privacy \citep{jewson2023beta}.

However, sampling from the DPD-based posterior continues to pose significant computational challenges, particularly for general parametric models. This is primarily owing to the presence of intractable integrals or infinite sums in the DPD formulation. Consequently, most studies on the DPD-based posterior have been limited to cases in which the normal distribution serves as a probability model. For general parametric models, although the intractable integral term can be theoretically approximated through numerical integration to enable posterior sampling {\citep{jewson2018principles, jackjewson2024StabilityGeneralBayesian}}, this approach is often inadequate for high-dimensional data \citep{niederreiter1992qmc}. This limitation is especially challenging given that high-dimensional datasets frequently contain outliers, as demonstrated in Outlier Detection Datasets \citep{Rayana:2016}, highlighting the need for robust DPD-based posterior sampling methods that perform well in high-dimensional settings.

To address these challenges, \cite{okuno2024minimizing} developed a novel stochastic gradient descent (SGD) algorithm for computing DPD-based estimators. Their approach significantly broadens the applicability of DPD-based estimation to encompass general parametric models, even in cases where the DPD lacks an explicit analytical form. This stochastic methodology offers substantial computational advantages over traditional gradient descent methods that use numerical integration, as it eliminates the need for precise integral approximations at each iteration. Although their framework guarantees theoretical convergence and maintains flexibility in a stochastic gradient design, it was primarily developed within a frequentist paradigm. 

Additionally, { intractable integrals of powered densities are often handled by framing them as expectations that can be estimated through sampling pseudo-observations. 
This is a central theme in likelihood-free inferences, where the likelihood is unavailable but simulating from the model is feasible. In this domain, \cite{pacchiardi2024general} proposed a generalized posterior distribution based on scoring rules and developed stochastic gradient Markov chain Monte Carlo (MCMC) methods for sampling, which scale effectively to high-dimensional parameter spaces without requiring summary statistics.  
Similarly, \cite{frazier2024loss} developed the first asymptotic theory for Gibbs measures in settings where the loss function must be inferred from pseudo-observations, which is a common challenge in generalized Bayesian and intractable-likelihood contexts.} Their findings illustrate that the accuracy of the resulting Gibbs measure critically depends on the bias and variance of the estimated loss. In particular, they reveal the number of pseudo-observations required to approximate the exact Gibbs measure, bridging an important theoretical gap in stochastic loss estimation. 
{They also propose a piecewise deterministic Markov process sampler for the Gibbs measure. However, these sampling methods need calibrations of a loss scaling parameter in generalized posterior distribution.}

Building on these developments, we propose {an approximate} sampling methodology that reconciles efficient computation with robust Bayesian inference. 
{Our approach combines the loss-likelihood bootstrap \citep[LLB;][]{lyddon2019general} with the SGD algorithm of \citet{okuno2024minimizing}, yielding a sampler applicable to a broad class of parametric models, including cases with intractable integrals. We further extend the method to generalized linear models (GLMs), a widely used class in practice. A key practical advantage is that the LLB-based sampler is readily parallelizable and requires no calibration of a loss scaling parameter $\omega$; in our formulation, the LLB-induced posterior random measure is invariant to $\omega$ and proposes a correct frequentist coverage \citep{lyddon2019general}. A limitation is that incorporating substantive prior information is nontrivial; we return to this point in the discussion and note connections to approaches that inject prior structure into LLB-type procedures \citep{fong2019posterior}. }Through extensive simulations, we demonstrate the computational efficiency of the method, particularly in high-dimensional settings.

The remainder of this paper is organized as follows. Section \ref{sec2} provides an overview of generalized Bayesian inference, introducing the theoretical foundations of the DPD-based posterior and the LLB. Section \ref{sec3} presents our novel methodology, detailing the stochastic gradient-based approach for {approximately} sampling from DPD-based posteriors, and its extension to GLMs. Section \ref{sec4} demonstrates the effectiveness of our proposed method through extensive simulation studies, including validation experiments, comparisons with conventional methods, and applications in intractable models. Finally, Section \ref{sec5} concludes the paper with a discussion of the findings and future research directions.

\section{Generalized Bayesian inference}
\label{sec2}
Here, we examine the methodological advancements in Bayesian inference under model misspecification and data contamination. First, we introduce the generalized posterior distribution, which extends traditional Bayesian methods by incorporating general loss functions and a loss scaling parameter. We then explore the DPD-based posterior as a robust approach for handling outliers, discussing both its theoretical guarantees and computational challenges. Finally, we present the LLB as an efficient method for generating approximate samples from generalized posterior distributions. These foundations serve as the basis for our proposed novel sampling method, which specifically addresses the computational challenges in DPD-based posterior inference.

\subsection{Generalized posterior distribution}
Let $\mathcal{X}\subset\mathbb{R}^d$ and $\Theta\subset\mathbb{R}^p$ denote sample and parameter spaces, respectively.
Suppose that $x_1, \dots, x_n\in \mathcal{X}$ are independent and identically distributed from distribution $F_0$ and $\pi(\theta)$ is a prior probability density function. 
Let $f(\cdot\mid\theta)$ ($\theta \in \Theta$) denote the statistical model on $\mathcal{X}$. 
While Bayesian inference provides methods to quantify the uncertainty in $\theta$, the true data-generating process $F_0$ is often unknown, and the assumed model may be misspecified.

To address this challenge, \cite{bissiri2016general} proposed a general framework for updating belief distributions using the generalized Bayesian inference. This framework enables quantification of the uncertainty of general parameter $\theta$, which has true value $\theta_0\coloneqq \argmin_{\theta\in\Theta}\int \ell(\theta, x)dF_0(x)$, where $\ell(\theta, x)$ is a general loss function. The generalized posterior distribution takes the form:
\[
\pi(\theta\mid x_{1:n})\propto\pi(\theta)\exp\{-\omega \ell(\theta, x_{1:n})\},
\]
where the loss function has an additive property, $\ell(\theta, x_{1:n})=\sum_{i=1}^n \ell(\theta, x_i)$ (here and throughout, we use the notation $x_{1:n}$ to denote $(x_1, \dots, x_n)$) and $\omega > 0$ is a loss-scaling parameter that controls the balance between the prior belief and the information from the data. If we consider $\omega = 1$ and $\ell(\theta, x) = -\log f(x\mid\theta)$, then we find that $\pi(\theta\mid x)\propto \pi(\theta)f(x\mid \theta)$, which shows that the generalized posterior distribution includes the standard posterior distribution as a special case. The selection of the loss scaling parameter $\omega$ has been extensively studied because it affects the posterior concentration and validity of the Bayesian inference \citep{petergrunwald2017InconsistencyBayesianInference, holmes2017Assigningvaluepower, lyddon2019general, syring2019Calibratinggeneralposterior, pei-shienwu2023ComparisonLearningRate}. 
{Additionally, foundational asymptotic theory for generalized posterior distributions (including concentration, Bernstein-von Mises normality, Laplace approximation, and frequentist coverage) can be found in \cite{miller2021general}.}

To achieve a robust Bayesian inference against outliers, we can construct a loss function using robust divergence measures. For example, using DPD or $\gamma$-divergence \citep{fujisawa2008robust} instead of the Kullback-Leibler divergence leads to inference procedures that are less sensitive to outliers while maintaining the theoretical guarantees of the generalized Bayesian inference \citep{ghosh2016robust, nakagawa2020robust}. These robust divergences provide principled ways to address the model misspecification and contaminated data while preserving the coherent updating property of Bayesian inference \citep{jewson2018principles}. 

\subsection{DPD-based posterior}
The DPD offers an attractive approach because of its robustness. The DPD-based loss function \citep{basu1998robust} is defined for observations $x_{1:n}$ as follows:
\begin{align}
    Q_n(\theta)&=\sum_{i=1}^n q(\theta, x_i),\notag\\
    q(\theta,x_i)&=-\frac{f(x_i\mid\theta)^\alpha}{\alpha}+\frac{1}{1+\alpha}\int_{\mathcal{X}}f(x\mid\theta)^\alpha dF_{\theta} (x),\label{eq:dpd_loss}
\end{align}
where $\alpha>0$ is a tuning parameter and $F_{\theta}$ is the distribution corresponding to $f(\cdot\mid\theta)$. The DPD-based posterior \citep{ghosh2016robust} provides a robust framework for updating prior distributions. The DPD-based posterior is defined as
\begin{align*}
    \pi_{\text{DPD}}(\theta\mid x_{1:n}) &\propto \pi(\theta)\exp\left\{-\omega Q_n(\theta)\right\}\\
    &=\pi(\theta)\exp\left\{-\omega \sum_{i=1}^n q(\theta, x_i)\right\}.
\end{align*} The DPD-based posterior inherits robust properties from the DPD loss function. Unlike traditional likelihood-based approaches, it is less sensitive to outliers and model misspecification. The robustness of the DPD-based posterior has been theoretically established through the analysis of influence functions \citep{ghosh2016robust}. The bounded influence function ensures protection against outliers and model misspecification while maintaining good efficiency under the true model. {Also, posterior predictive distributions via DPD-based posterior has predictive stability to model's specification better than standard posterior predictive distribution \citep{jackjewson2024StabilityGeneralBayesian}.}

{Tuning parameter $\alpha$ controls the trade-off between efficiency and robustness: larger $\alpha$ yields greater robustness at the cost of statistical efficiency under correct specification. Methods for choosing $\alpha$ have been developed in the frequentist DPD literature \citep{warwick2005tuning, basak2021tuning} and in DPD-based posteriors; in particular, \citet{yonekura2023adaptation} proposed a data‐adaptive selection of $\alpha$ via Hyv\"arinen score. Together, these works offer theoretical and practical guidance for setting $\alpha$ according to data characteristics and application context.}

The DPD-based posterior is highly desirable because it is robust against outliers. However, the second term of the loss function \eqref{eq:dpd_loss} based on DPD contains an integral term (infinite sum) that cannot be written explicitly, except for specific distributions, such as the normal distribution, making it difficult to calculate the posterior distribution using MCMC. Although numerical integration can be used to approximate the calculations in general probability models, it often fails for high-dimensional data \citep{niederreiter1992qmc}. To address these computational challenges, we propose a new {approximate} sampling method that avoids the need for an explicit calculation or numerical approximation of the integral term. Before presenting our proposed approach, we first review the LLB introduced by \cite{lyddon2019general}, which provides an important theoretical foundation for our sampling algorithm. The LLB offers a nonparametric method for approximating posterior distributions by resampling; however, it has limitations when applied directly to DPD-based inference. We then build on these ideas to develop an efficient sampling scheme specifically designed for DPD-based posteriors using general parametric models.
\subsection{Loss-likelihood bootstrap}
{
Motivated by the need to retain uncertainty quantification without resorting to costly MCMC, and especially because generalized posterior distributions induced by non-likelihood losses often exhibit complex geometry (e.g., non-log-concavity and multimodality) that makes MCMC delicate, we begin by revisiting the weighted likelihood bootstrap \citep[WLB;][]{newton1994approximate} as an optimization-based, trivially parallelizable surrogate for likelihood-based posterior sampling.
In the WLB, draw Dirichlet weights $n^{-1}(w_1,\ldots,w_n)^\top\sim\mathrm{Dir}(1,\ldots,1)$ and form the random empirical measure $F=\sum_{i=1}^n w_i\delta_{x_i}$, where $\delta_{x_i}$ denotes the unit point mass at $x_i$. Propagate this randomness through the self-information loss $\ell_{\mathrm{si}}(\theta,x)=-\log f(x\mid\theta)$ by computing 
$\theta(F)=\argmin\int \ell_{{\mathrm{si}}}(\theta,x)dF(x).$
Thus, uncertainty about the unknown $F_0$ is represented by Dirichlet perturbations on the empirical support, and uncertainty about $\theta$ is obtained by solving a deterministic optimization problem for each random reweighting rather than by running a Markov chain. \cite{lyddon2018nonpara} also extends this methodology to nonparametric Bayesian learning called as posterior bootstrap. This viewpoint provides a seamless bridge from likelihood-based Bayesian updating to the generalized Bayes setting: replacing $\ell_{\mathrm{si}}$ in the WLB objective with an arbitrary additive loss $\ell$ yields LLB \citep{lyddon2019general}, which preserves the computational virtues of WLB (independent, easily parallelized draws) while enabling the calibrated generalized posterior distribution through information matching.
}

LLB \citep{lyddon2019general} provides a computationally efficient method for generating approximate samples from generalized posterior distributions by leveraging Bayesian nonparametric principles. The key insight is to model the uncertainty of the unknown true data-generating distribution using random weights in the empirical distribution. The algorithm proceeds as follows:

\begin{algorithm}[H]
\caption{loss-likelihood bootstrap}
\label{alg:llb}
\begin{algorithmic}[1]
\For{$s = 1,\ldots,S$}
    \State Generate Dirichlet weights $(w_{s1},\ldots,w_{sn})^{\top} \sim \text{Dir}(1,\ldots,1)$
    \State Set empirical measure $F^{(s)} = \sum_{i=1}^n w_{si}\delta_{x_i}$ 
    \State Compute $\theta^{(s)} = \argmin_{\theta}\int \ell(\theta,x)dF^{(s)}(x)$
\EndFor
\Return ${\{}\theta^{(1)},\ldots,\theta^{(S)}{\}}$
\end{algorithmic}
\end{algorithm}

We characterize the asymptotic properties of the LLB samples using Theorem \ref{thm1} \citep{lyddon2019general}. This result indicates that the LLB samples exhibit {proper uncertainty quantification asymptotically, in the sense that the resulting credible intervals achieve correct frequentist coverage,} with sandwich covariance matrix $J(\theta_0)^{-1}I(\theta_0)J(\theta_0)^{-1}$ providing robustness to model misspecification.

\begin{theorem}[\cite{lyddon2019general}]
\label{thm1}
Let $\tilde\theta_{n}$ be the proposed LLB sample of parameter $\theta_0$ with loss function $\ell(\theta, \cdot)$ given $n$ observations $x_{1:n}$, and let $P_{\text{LLB}}$ be its probability measure. Under suitable regularity conditions, for any Borel set $A \subset \mathbb{R}^p$, as $n \to \infty$:
$$P_{\text{LLB}}(n^{1/2}(\tilde{\theta}_n - \hat{\theta}_n) \in A \mid x_{1:n}) \stackrel{a.s.}{\longrightarrow} P(z \in A),$$
where $\hat{\theta}_n = \argmin_{\theta} n^{-1}\sum_{i=1}^n \ell(\theta,x_i)$ is the empirical risk minimizer, $z$ follows normal distribution $N_p(0, J(\theta_0)^{-1}I(\theta_0)J(\theta_0)^{-1})$ with $I(\theta) = \int \nabla_\theta\ell(\theta,x)\nabla_\theta\ell(\theta,x)^\top dF_0(x)$ and $J(\theta) = \int \nabla_\theta^2\ell(\theta, x) dF_0(x)$, and $\nabla_\theta$ denotes the gradient operator with respect to $\theta$.
\end{theorem}

The LLB approach offers significant computational advantages over the traditional MCMC methods. {By generating independent samples and enabling trivial parallelization, it does not require MCMC convergence diagnostics such as effective sample size or R-hat \citep{lyddon2019general}.
However, a known limitation is that it is difficult to incorporate prior information directly into the LLB framework. Recent work by \cite{fong2019posterior} addresses this by introducing a Dirichlet process prior on the sampling distribution and by allowing a loss-based penalization to encode prior beliefs while retaining the computational advantages of the posterior bootstrap. Additionally, \citet{newton2021weighted} incorporate prior information into the weighted Bayesian bootstrap by adding a log-prior penalty to the randomized objective function.
}

Moreover, the LLB plays a crucial role in calibrating the loss scale parameter for the generalized posterior distributions. Under regularity conditions, both the generalized posterior and LLB posterior distributions converge to normal distributions with identical means at empirical risk minimizer $\hat{\theta}_n$ but different covariance structures: $w^{-1}J(\theta_0)^{-1}$ for the generalized posterior and $J(\theta_0)^{-1}I(\theta_0)J(\theta_0)^{-1}$ for the LLB posterior. By equating the Fisher information numbers \citep{walker2016bayes}, defined as $K(p) = \text{tr}(\Sigma^{-1})$ for a normal distribution with covariance $\Sigma$, they derive the calibration formula $w = \text{tr}\{J(\theta_0)I(\theta_0)^{-1}J(\theta_0)^\top\}/\text{tr}\{J(\theta_0)\}$. This calibration exhibits a remarkable theoretical property: when the true data-generating distribution belongs to the model family up to a scale factor (i.e., $f_0(x) = \exp\{-w_0\ell(\theta_0,x)\}$ for some $\theta_0 \in \Theta$ and $w_0 > 0$), it recovers the correct scale $w = w_0$. In practice, because $\theta_0$ is unknown, the matrices $I$ and $J$ are estimated empirically using the risk minimizer $\hat{\theta}_n$, resulting in a computationally tractable approach for calibrating the generalized posterior that ensures equivalent asymptotic information content to the LLB posterior. This theoretical result, established by \cite{lyddon2019general}, provides further justification for viewing the LLB as an approximate sampling method from a well-calibrated generalized posterior distribution.

\section{Methodology}
\label{sec3}
It is crucial for robust Bayesian inference to sample from DPD-based posterior for a general parametric model; however, its implementation has been computationally challenging owing to the integral terms involved in the DPD calculation. Although traditional approaches have been limited to specific probability models where explicit solutions are available, recent advances in stochastic optimization techniques have opened new possibilities for efficient sampling methods. Here, we propose a novel sampling algorithm that leverages these advances to enable efficient posterior sampling. Our approach reformulates the sampling procedure using stochastic gradient techniques, making it computationally tractable for a broad class of statistical models, including GLMs.
\subsection{New sampling algorithm from DPD-based posterior}
Recently, \cite{okuno2024minimizing} achieved a significant breakthrough in minimizing the DPD for general parametric density models by introducing a stochastic optimization approach. Whereas traditional DPD minimization requires  the computation of integral terms that can only be explicitly derived for specific distributions, such as normal and exponential distributions, their stochastic optimization framework enables DPD minimization for a much broader class of probability density models.

Inspired by this stochastic approach to minimize DPD, we propose reformulating the LLB algorithm {\citep{lyddon2019general, dellaporta2022robust}} as a stochastic method. While the LLB was traditionally formulated as an optimization problem based on weighted loss-likelihoods, the introduction of stochastic optimization techniques enables the construction of a more computationally efficient algorithm.

The objective function, $L_w(\theta)$, that we minimize in the LLB algorithm, is 
\begin{align*}
    &\mathcal{L}_{w}(\theta)=\sum_{i=1}^n w_i q(\theta, x_i),\label{eq2}\\
    &(w_1,\dots, w_n)^{\top} \sim \text{Dir}(1, \dots, 1). 
\end{align*}
We also calculate the gradient of the objective function as
\begin{align*}
    \nabla_\theta\mathcal{L}_w(\theta)&=\nabla_\theta\Biggr[\sum_{i=1}^n w_i q(\theta, x_i)\Biggr] \\
    &=-\sum_{i=1}^n w_i f(x_i\mid\theta)^\alpha u(x_i\mid\theta) 
    +E_{ F_\theta}\Bigr[f(X\mid\theta)^\alpha u(X\mid\theta)\Bigr],
\end{align*}
where $ u(\cdot\mid\theta)=\nabla_\theta\log f(\cdot\mid\theta)$. 
We then randomly generate independent samples $\xi^{(t)}={\{}z_1^{(t)},\dots, z_m^{(t)}{\}}$ from $F_{\theta^{(t)}}$ and define the stochastic gradient of step $t$ as
\begin{equation}
  g(\theta^{(t)}\mid\xi ^{(t)})=-\sum_{i=1}^n w_i f(x_i\mid\theta^{(t)})^\alpha u(x_i\mid\theta^{(t)})+\frac{1}{m}\sum_{j=1}^m f(z_j^{(t)}\mid\theta^{(t)})^\alpha u(z_j^{(t)}\mid\theta^{(t)}).
  \label{eq3}
\end{equation}
Let $v > 0$ and assume that the objective function $\mathcal{L}_w$ satisfies the following conditions:
(i) $\mathcal{L}_w(\theta)$ is differentiable over $\Theta$
(ii) There exists $L > 0$ such that for any $\theta, \theta' \in \Theta$,
$\|\nabla_\theta \mathcal{L}_w(\theta) - \nabla_\theta \mathcal{L}_w(\theta')\| \leq L\|\theta - \theta'\|$
(iii) $\E_{\xi^{(t)}}[{g(\theta^{(t)}\mid\xi^{(t)})}] = \nabla_\theta\mathcal{L}_w(\theta^{(t)})$ for all $t = 1,\ldots,T$
(iv) ${\E_{\xi^{(t)}}[\|g(\theta^{(t)}\mid\xi^{(t)})-\nabla_\theta\mathcal{L}_w(\theta^{(t)})\|^2]} \leq v$ for all $t = 1,\ldots,T$.
Under these conditions, as $T \to \infty$, the SGD algorithm converges to minimize the objective function $\mathcal{L}_w(\theta)$ in probability \citep{ghadimi2013sgd, okuno2024minimizing}. 
{Furthermore, in the Supplementary Materials, for regular exponential families under mild compactness and envelope conditions, assumptions (ii)–(iv) are satisfied.
}
The parameter $\theta^{(t)}$ was iteratively updated using the SGD method as follows: 
\begin{equation*}
\theta^{(t)}\leftarrow\theta^{(t-1)}-\eta_{t-1}g(\theta^{(t-1)}\mid\xi^{(t-1)}) ~ (t=1, \dots, T).
\end{equation*}
{We assume the learning rates $\{\eta_t\}_{t=1}^T$ satisfy the following conditions \citep{ghadimi2013sgd, okuno2024minimizing}:}
\begin{equation*}
\eta_t < 2/L, ~\sum_{t=1}^T\eta_t = \infty, ~\Biggr\{\sum_{t=1}^T\eta_t\Biggr\}^{-1}\Biggr\{\sum_{t=1}^T\eta_t^2\Biggr\}=0, ~T\rightarrow\infty.
\end{equation*} 

{Here, the number of pseudo-samples $m$ does not theoretically affect minimization of the randomized objective function $L_w$, even when $m=1$; \citep[see][]{ghadimi2013sgd, okuno2024minimizing}. In our simulation studies reported in the Supplementary Materials, the proposed method exhibited stable performance across a wide range of $m$. For numerical stability in practice, we recommend setting $m$ more than $10$.
}

Based on these theoretical foundations, we implement our approach using the SGD algorithm that offers two key advantages. First, it provides a unified framework that enables sampling from a broad class of parametric probability models, thereby overcoming the traditional limitations of explicit DPD computations. This generality renders the method applicable across diverse statistical scenarios without requiring model-specific derivations. Second, as demonstrated in subsequent simulation studies, the algorithm exhibits substantial computational efficiency even in high-dimensional settings, where traditional sampling methods often become intractable. 
{We emphasize that the procedure delivers approximate sampling rather than exact draws from a DPD-based posterior. In particular, for finite observation size $n$, the LLB samples can be understood as an estimator of the sampling distribution of the minimum DPD estimator.}
We formalize this approach in Algorithm \ref{alg:proposed}, which details the implementation of our stochastic gradient-based sampling procedure.

\begin{algorithm}[H]
\caption{LLB with SGD}
\label{alg:proposed}
\begin{algorithmic}[1]
\Require Data points $x_1,\ldots,x_n$, initial parameter $\theta_{\text{init}}$, number of bootstrap samples $S$, maximum iterations $T$, initial learning rate $\eta_{\text{init}}$
\Ensure Bootstrap samples $\theta_1,\ldots,\theta_S$

\For{$s=1,\dots,S$}
   \State $(w_1,\ldots,w_n)^{\top} \sim \text{Dir}(1,\ldots,1)$ 
   \State $\theta \leftarrow \theta_{\text{init}}, ~ \eta \leftarrow \eta_{\text{init}}$
   
   \For{$t=1,\dots,T$}
       \State Generate $m$ random samples {$z_1^{(t)},\ldots,z_m^{(t)}$ from $F_{\theta^{(t)}}$}
       
       \State {$\theta^{(t)} \leftarrow \theta^{(t-1)} - \eta_{t-1} g(\theta^{(t-1)}\mid \xi^{(t-1)})$}
       \State Periodically update learning rate $\eta_{t}$
   \EndFor
   \State {$\theta_s \leftarrow \theta^{(T)}$ }
\EndFor
\Return ${\{}\theta_1,\dots,\theta_S{\}}$
\end{algorithmic}
\end{algorithm}

\subsection{Sampling from DPD-based posterior in GLMs}
Building on the stochastic gradient-based sampling framework introduced in the previous section, we {extend} our methodology to include GLMs. {GLMs entail observation-specific expectation/integral terms that depend on the predictors $x_i$, so one must approximate $n$ distinct integrals/sums, one per observation, at each iteration.}
As in the previous section, our approach minimizes the objective function through SGD, and all theoretical properties, including the convergence conditions (i)–(iv), remain identical. This theoretical consistency allows us to leverage the same optimization framework {while explicitly handling the observation-specific integral structure required by GLMs}.

In GLMs, the random variables $y_i~(i=1, \dots, n)$ are independent and follow a general exponential family of distributions with a density
\begin{equation*}
    f(y_i\mid\theta_i, \phi) = \exp\Biggr\{\frac{y_i\theta_i-b(\theta_i)}{a(\phi)}+c(y_i,\phi)\Biggr\},
\end{equation*}
where canonical parameter $\theta_i$ is a measure of the location depending on fixed predictor $x_i$ and $\phi$ is the scale parameter. Let $\mu_i=E[Y_i\mid x_i, \beta]$; thus, the GLM satisfies $h(\mu_i)=x_i^\top \beta$, where $h(\cdot)$ is a monotone and differentiable function.

Applying our stochastic gradient-based approach to GLMs, we consider Bayesian inference via DPD for the regression coefficients and scale parameters. Let $\theta = (\beta^{\top}, \phi)^{\top}$ and let $\alpha > 0$ be a hyperparameter, we can obtain the following loss function:
\begin{align*}
    Q_n(\theta) &= \sum_{i=1}^n q(\theta,(y_i, x_i)),\\
    q(\theta, (y_i, x_i)) &= -\frac{1}{\alpha}f(y_i\mid x_i,\theta)^\alpha+\frac{1}{1+\alpha}\int f(y\mid x_i, \theta)^{\alpha} dF_{i{, \theta }}(y).\label{eq:dpd_loss_glm}
\end{align*}

{Because $F_{i,\theta}$ depends on $x_i$, the integral in $q(\cdot)$ differs for each $i$, so we must approximate $n$ distinct expectation terms. Moreover, in the weighted objective $\mathcal{L}_w(\theta)=\sum_{i=1}^n w_i q(\theta,(y_i,x_i))$, the same bootstrap weight $w_i$ multiplies both the data-fit term and the integral/expectation term for each observation, which is a subtle but important point that is not required when a single shared integral is used.} 
Following the same principles as our general framework, we derive the stochastic gradient estimator for GLMs. The gradient of the objective function takes the following form:
\begin{align*}
    \nabla_\theta\mathcal{L}_w(\theta)&=\nabla_\theta\Biggr[\sum_{i=1}^n w_i q(\theta, (y_i, x_i))\Biggr] \\
    &=\sum_{i=1}^n w_i\Biggr[-f(y_i\mid x_i, \theta)^\alpha u(y_i\mid x_i, \theta) 
    +E_{F_{i{, \theta }}}\Bigr[f(Y\mid x_i, \theta)^\alpha u(Y\mid x_i, \theta)\Bigr]\Biggr],
\end{align*}
where $u(\cdot\mid x_i, \theta)=\nabla_\theta \log f(\cdot\mid x_i, \theta)$ and $F_{i{, \theta }}$ is the distribution corresponding to $f(y_i\mid x_i, \theta)$. We randomly generate independent samples $\xi_i^{(t)}={\{}z_{i1}^{{(t)}}, \dots, z_{im}^{{(t)}}{\}}$ from distributions $F_{i{, \theta }}$ and $\xi^{(t)}={\{}\xi_1^{(t)}, \dots, \xi_n^{(t)}{\}}$. We can obtain the gradient of the objective function such that
\begin{align*}
      g(\theta^{(t)}\mid\xi^{(t)})&=\sum_{i=1}^n g(\theta^{(t)}\mid\xi_i^{(t)}),\\
      g(\theta^{(t)}\mid\xi_i^{(t)}) &= w_i\Bigr\{-f(y_i\mid x_i, \theta^{(t)})^\alpha u(y_i\mid x_i,\theta^{(t)})+\frac{1}{m}\sum_{j=1}^m f(z_{ij}\mid x_i,\theta^{(t)})^{\alpha}u(z_{ij}\mid x_i,\theta^{(t)})\Bigr\}.
\end{align*}

\begin{algorithm}[H]
\caption{LLB with SGD for GLMs}
\label{alg:proposed-glm}
\begin{algorithmic}[1]
\Require Data points $(y_1,  x_1),\ldots,(y_n,  x_n)$, initial parameter $\theta_{\text{init}}$, number of bootstrap samples $S$, maximum iterations $T$, initial learning rate $\eta_{\text{init}}$
\Ensure Bootstrap samples $\theta_1,\ldots,\theta_S$

\For{$s=1,\dots,S$}
   \State $(w_1,\ldots,w_n)^{\top} \sim \text{Dir}(1,\ldots,1)$ 
   \State $\theta \leftarrow \theta_{\text{init}},~\eta \leftarrow \eta_{\text{init}}$
   
   \For{$t=1,\dots,T$}
       \For{$i=1,\dots,n$}
           \State Generate samples $z_{i1}^{{(t-1)}},\ldots,z_{im}^{{(t-1)}}$ from $F_{i, \theta^{(t-1)}}$
           \State Compute component gradient: {$g(\theta^{(t-1)}\mid\xi_i^{(t-1)})$} 
       \EndFor
       \State Compute full gradient: ${g(\theta^{(t-1)}\mid\xi^{(t-1)})} \leftarrow \sum_{i=1}^n {g(\theta^{(t-1)}\mid\xi_i^{(t-1)})}$
       \State {$\theta^{(t)} \leftarrow \theta^{(t-1)} - \eta_{t-1} g(\theta^{(t-1)}\mid \xi^{(t-1)})$}
       \State Periodically update learning rate $\eta_{t}$
   \EndFor
   \State {$\theta_s \leftarrow \theta^{(T)}$}
\EndFor
\Return ${\{}\theta_1,\dots,\theta_S{\}}$
\end{algorithmic}
\end{algorithm}

{
This work extends our DPD-based posterior sampling framework to GLMs. In this context, the gradient calculation requires handling $n$ observation-specific expectation terms, and the procedural structure necessitates that each of these $n$ terms be individually approximated. Our primary contribution lies in the method of this approximation: instead of employing deterministic methods like numerical integration for each term \citep{jewson2018principles, jackjewson2024StabilityGeneralBayesian}, we construct a computationally efficient, stochastic unbiased estimator for each expectation using pseudo-samples. This specific substitution is what enables efficient and stable posterior sampling for various GLMs, including Poisson regression. Algorithm \ref{alg:proposed-glm} provides the implementation details for this procedure, specifying the loop through the $n$ observations and, for each one, detailing the step of forming the stochastic gradient component from pseudo-samples. This approach preserves the core iterative structure presented in Algorithm \ref{alg:proposed}.
}

\section{Simulation studies}
\label{sec4}
Here, we first verify that our proposed method accurately samples from the DPD-based posterior and benchmark it against conventional alternatives to assess accuracy and efficiency. 
{We report a sensitivity analysis of the robustness parameter $\alpha$ in the Supplementary Materials (Section C), showing that our qualitative conclusions are stable over a reasonable range of $\alpha$; accordingly, we fix $\alpha=0.5$ in all subsequent simulations. We also provide in the Supplementary Materials (Section D) a direct comparison with the standard Bayesian posterior under $\varepsilon$-contamination, as well as a two-dimensional case study in which LLB-based sampling is compared with MCMC for the same DPD-based posterior.}
The R code used to reproduce the experimental results is available at \url{https://github.com/naruki-sonobe/DPD-Bayes-SGD}, and additional simulation results are provided in the Supplementary Materials {(Sections B-C)}.
\subsection{Validation of the proposed method through simulations}
We evaluated our proposed LLB algorithm, which replaces the gradient of the objective function with an unbiased stochastic gradient, through two comprehensive numerical experiments. For our evaluation dataset, we generated $1,000$ samples from a normal distribution $\text{N}(0, 1)$ and introduced outliers from $\text{N}(10, 0.01)$ at a contamination ratio of $0.05$. We assumed a normal distribution with parameters $(\mu, \sigma)$ as the probability model. This setting is particularly advantageous because it allows us to compute the integral term of \eqref{eq:dpd_loss} explicitly as $(2\pi)^{-\frac{\alpha}{2}}(1+\alpha)^{-\frac{3}{2}}\sigma^{-\alpha}$, enabling us to express the DPD-based posterior in closed form. This tractability facilitates the application of various sampling methods, including MCMC with the MH algorithm and the standard LLB algorithm, providing an ideal framework for comparative analysis. {Here, the standard LLB algorithm is a sampling method based on Algorithm \ref{alg:llb} that leverages analytically derived gradients.}
{In this section, we fix the number of pseudo-samples in the proposed method at $m=100$. For the MCMC implementation, the loss scaling parameter of the DPD-based posterior is calibrated via information matching \citep{lyddon2019general}.}

{In our first experiment, the stochastic gradient closely approximates the true gradient of the objective function in the LLB algorithm. Across 20 Monte Carlo experiments (10,000 posterior draws per run), Figure~\ref{fig_exact_LLB} overlays the 20 kernel density estimate (KDE) curves for both methods. These results demonstrate that the performance of the LLB algorithm remains consistent regardless of the stochastic gradient implementation. Consistently, Table~\ref{tab:avg_post} show that LLB and LLB with SGD yield nearly identical means and variances for both $\mu$ and $\sigma$, corroborating this agreement.
 
Our second experiment evaluates whether our method accurately samples from the DPD-based posterior by comparing it with Metropolis-Hastings (MH) and Hamiltonian Monte Carlo (HMC, via Stan \citep{carpenter2017stan}). Figure~\ref{fig_mcmc} overlays the 20 KDE curves for each method (10,000 draws per run). The near-coincidence of these distributions confirms that our method samples from the target DPD-based posterior. Likewise, Table~\ref{tab:avg_post} reports closely matching posterior means and variances across MH, HMC, and our method, reinforcing that all approaches target the same DPD‐based posterior.
\begin{table}[h]
\centering
\caption{Average posterior mean and variance\\over 20 Monte Carlo experiments}
\label{tab:avg_post}
\begin{tabular}{rrrr}
\hline
Method & Parameter & Mean & Variance \\
\hline
LLB & $\mu$     & 0.0029 & 0.0012 \\
    & $\sigma$ & 1.0130 & 0.0008 \\
LLB with SGD & $\mu$     & 0.0031 & 0.0013 \\
             & $\sigma$ & 1.0133 & 0.0008 \\
MH  & $\mu$     & $-0.0025$ & $0.0011$ \\
    & $\sigma$ & $1.0133$  & $0.0008$ \\
HMC & $\mu$     & $-0.0024$ & $0.0011$ \\
    & $\sigma$ & $1.0140$  & $0.0008$ \\
\hline
\end{tabular}
\end{table}
}
\begin{figure}[H]
    \centering
    \includegraphics[width=0.8\textwidth]{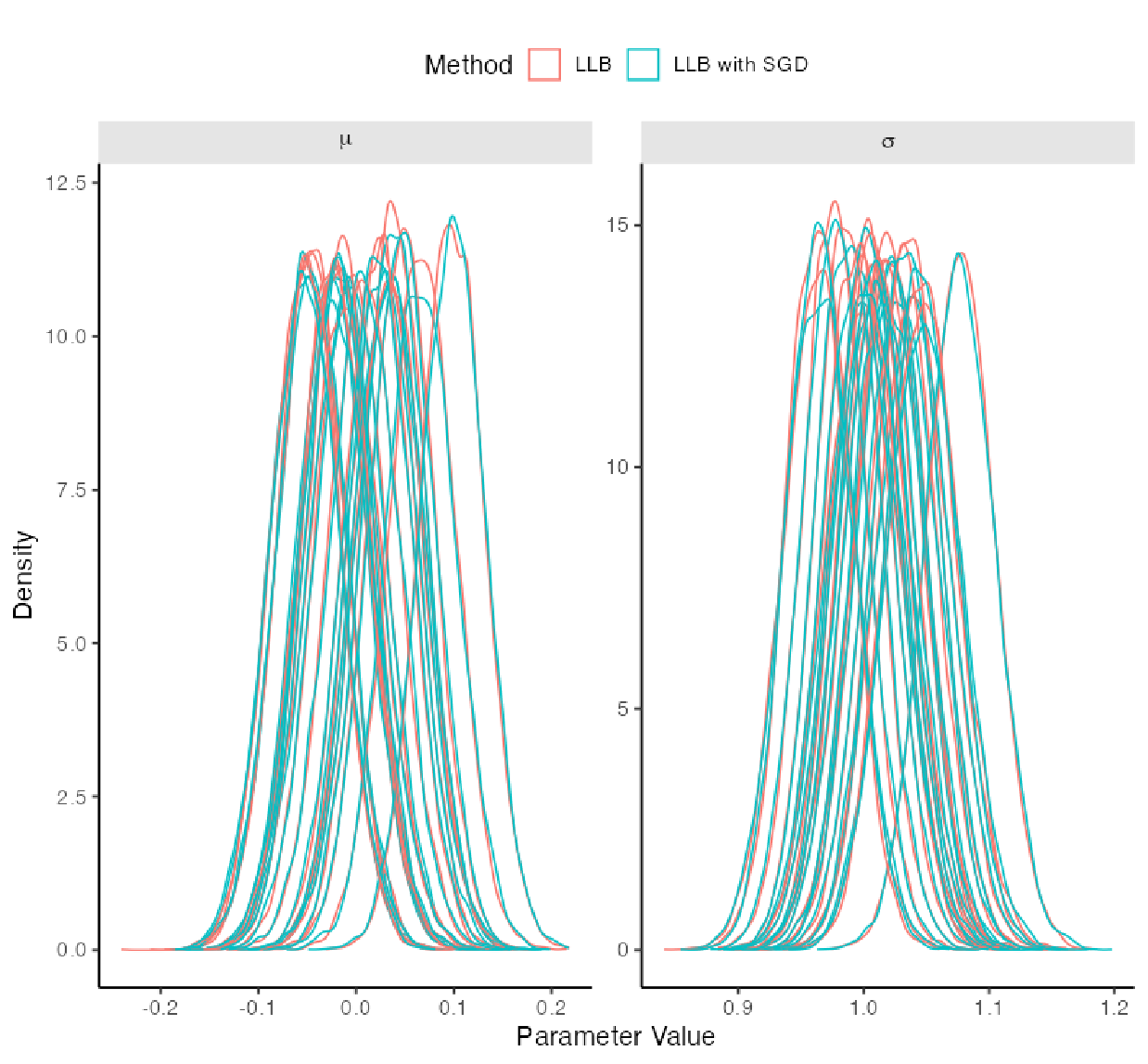}
    \vspace{-2.5mm}
    \caption{Comparison with the standard LLB {over 20 Monte Carlo experiments}}
    \label{fig_exact_LLB}
    \centering
    \includegraphics[width=0.8\textwidth]{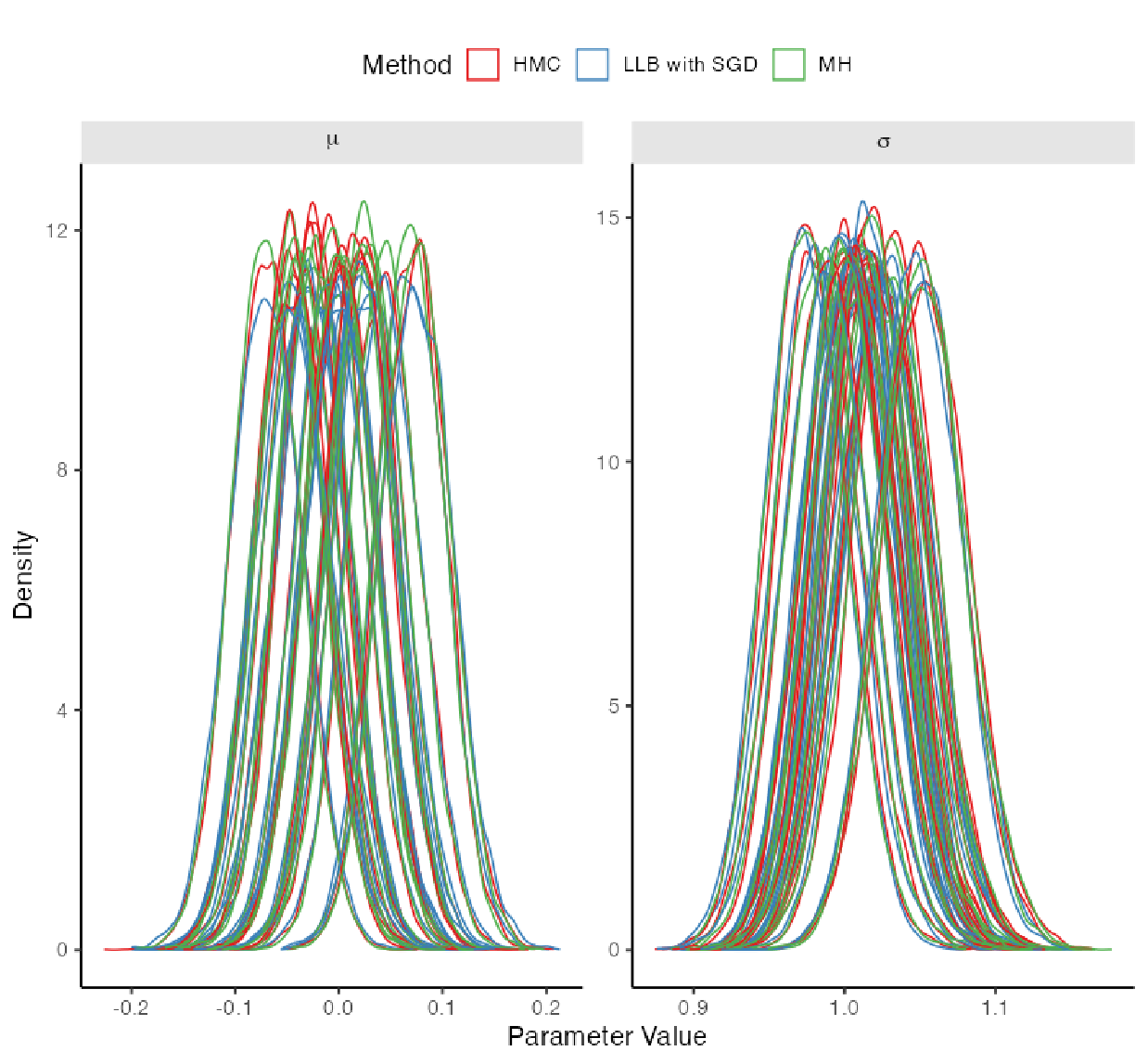}
    \vspace{-2.5mm}
    \caption{Comparison with the standard LLB {over 20 Monte Carlo experiments}}
    \label{fig_mcmc}
\end{figure}

\subsection{Comparison of computational efficiency with conventional methods}
Conventional sampling methods often encounter significant computational challenges. These challenges primarily manifest as in increased computational time and reduced estimation accuracy. To validate the scalability and effectiveness of the proposed method for addressing these challenges, we conducted a comprehensive simulation study.

To evaluate the validity of our proposed method in high-dimensional settings, we designed our simulation framework using a contaminated multivariate normal model, in which we generated $100$ samples from $N_p(\mathbf{0}, \mathbb{I})$ with $5\%$ contamination from $N_p(10\ind, 0.01 \mathbb{I})$. 
{Here, $\mathbb{I}$ denotes the identity matrix, and $\ind$ represents the vector of ones, and $0$ denotes the zero vector. 
For our proposed method, we set the number of pseudo-samples in proposed method $m=10$.}
For the numerical integration in the LLB algorithm, we performed the integration over a grid of $M^p$ points spanning the hypercube $[-D, D]^p$, where we set $M=10$ grid points per dimension and $D=2$ as the integration boundary. While this grid-based numerical integration approach provides high accuracy for one-dimensional problems, the computational cost grows exponentially with the dimensions, as it requires $M^p$ evaluations for $p$-dimensional integration. To optimize the LLB, we used a standard gradient descent with a fixed learning rate $\eta=T^{-1}\sum_{t=1}^T \eta_t$. This setting allowed us to systematically evaluate the effects of increasing dimensionality. We compared our proposed method with the LLB algorithm with numerical integration, both initialized with maximum likelihood estimates and configured to generate $1,000$ posterior samples.

Table \ref{tab:comparison} compares the computational efficiency between our SGD approach and the LLB method using gradient descent with numerical integration (GD+NI) across varying dimensions. {{Across 1,000 Monte Carlo replications, both methods achieved similar mean squared error and posterior variance for $p=2,3$}, indicating comparable estimation accuracy at low dimensions.} The key difference lies in the computational time required to obtain 1,000 samples. 
{In terms of runtime, SGD remains nearly constant at 4.6–5.3 s to obtain 1,000 samples as the dimension increases. By contrast, GD+NI grows rapidly: it requires 6.5 s at $p=2$ and 8.7 s at $p=3$, but then jumps to 526.7 s at $p=5$; at $p=7$ the routine stalled (no convergence within the time limit), so no value is reported. Accordingly, we omit the GD+NI timing at $p=7$ and, to avoid an uneven comparison, we report accuracy metrics (MSE and posterior variance) only for $p=2,3$ and omit them for $p\ge 5$.
This marked divergence in computational time, while maintaining similar accuracy where both methods are stable, highlights the superior scalability of our stochastic approach in higher dimensions.
}

\begin{table}[h]
\centering
\caption{Comparison between SGD and GD+NI methods for obtaining 1,000 samples across different dimensions}
\label{tab:comparison}
\begin{tabular}{rrrrr}
\hline
 & & Time (Seconds) & Mean squared error & Posterior variance\\
\hline
SGD   & $p = 2$ & 4.6770 & 0.0135 & 0.0142\\
      & $p = 3$ & 4.6990 & 0.0134 & 0.0149\\
      & $p = 5$ & 4.8060 & - & -\\
      & $p = 7$ & 5.2220 & - & -\\
GD+NI & $p = 2$ & 6.5190 & 0.0157 & 0.0160\\
      & $p = 3$ & 8.7060 & 0.0153 & 0.0168\\
      & $p = 5$ & 526.7240 & - & -\\
      & $p = 7$ & - & - & -\\
\hline
\end{tabular}
\end{table}
These empirical results not only support our theoretical findings but also demonstrate the practical viability of our method for high-dimensional robust Bayesian inference. The ability to maintain computational efficiency without sacrificing estimation accuracy makes our method particularly suitable for modern high-dimensional applications where traditional numerical integration approaches become computationally prohibitive.

\subsection{Numerical experiments with intractable models}
To validate the practical utility of our method in scenarios in which traditional approaches face computational challenges, we conducted extensive simulation studies using Poisson regression models. This choice of model is particularly important because it represents a common case in which the integral term in \eqref{eq:dpd_loss} lacks a closed-form solution, thus requiring a numerical approximation in conventional approaches. 
{We show the validity of the proposed method through both synthetic and real data.}

\subsubsection{Synthetic data analysis}
For our simulation setup, we generated synthetic data comprising $n=300$ samples from a Poisson regression model, and the experiment was repeated 100 times. The covariates $ x_i \in \mathbb{R}^p$ were drawn from a standard multivariate normal distribution $N_p(\mathbf{0}, \mathbb{I})$, and the response variables were generated according to $y_i \sim \text{Po}(\exp(\beta_0+ x_i^\top \beta))$. We construct the true parameter $(\beta_0, \beta)$ was constructed by independently sampling each component from a uniform distribution over $[0, 1/4]$, thereby ensuring a realistic but challenging estimation scenario.

We compared the proposed stochastic optimization approach with the standard LLB algorithm, which relies on a finite-sum approximation of the infinite sum. For the existing method, we employed the same gradient descent algorithm, as described in the previous subsection. Following the established practice in the literature \citep{kawashima2019robust}, the LLB implementation uses $M=n$ Monte Carlo samples for approximation.
{For our proposed method, we fixed the number of pseudo-samples at $m=n$ similarly in all runs.}
Both methods were initialized using maximum likelihood estimates and configured to generate $1,000$ posterior samples to ensure a fair comparison of their performance.

{Table \ref{tab:performance} presents a comprehensive comparison across varying dimensions ($p=2,5,10,20$) using three key performance metrics averaged over 100 simulation runs: the mean squared error, coverage probability, and average length of 95\% credible intervals. To evaluate the performance, we first computed the posterior median $\hat{\beta}_k$ for each regression coefficient $\beta_k$ and assessed their accuracy using the mean squared error, which is calculated as $(p+1)^{-1}\sum_{k=0}^p(\hat{\beta}_k - \beta_k)^2$. Additionally, we constructed 95\% credible intervals $CI_k$ for each $\beta_k$ and evaluated their effectiveness using two measures: the average length of the intervals, given by $(p+1)^{-1}\sum_{k=0}^p |CI_k|$, and the coverage probability, defined as $(p+1)^{-1}\sum_{k=0}^p I(\beta_k \in CI_k)$.}

Our proposed method shows consistently superior performance across all dimensions, not only achieving lower MSE values but also maintaining coverage probabilities closer to the nominal 95\% level. Furthermore, it provides more precise uncertainty quantification, as demonstrated by the shorter average interval lengths, without compromising coverage. These results convincingly demonstrate that our stochastic optimization approach delivers a more accurate and reliable inference than traditional methods, particularly in settings where model complexity precludes analytical evaluation of the objective function. The robust performance across increasing dimensions further highlights the scalability of our approach, which makes it particularly attractive for modern high-dimensional applications.

\begin{table}[h]
\centering
\caption{Comparison of the estimation performance between the proposed and existing methods}
\label{tab:performance}
\begin{tabular}{rrrrrr}
\hline
Performance Measure & Method & $p = 2$ & $p = 5$ & $p = 10$ & $p = 20$ \\
\hline
\multirow{2}{*}{Mean Squared Error} 
& Proposed & 0.0037 & 0.0035 & 0.0037 & 0.0034 \\
& Existing & 0.2011 & 0.1353 & 0.0763 & 0.0431 \\
\hline
\multirow{2}{*}{Coverage Probability} 
& Proposed & 0.9233 & 0.9433 & 0.9409 & 0.9543 \\
& Existing & 0.7967 & 0.7083 & 0.7391 & 0.8124 \\
\hline
\multirow{2}{*}{Average Length} 
& Proposed & 0.2343 & 0.2342 & 0.2317 & 0.2329 \\
& Existing & 1.3740 & 0.8711 & 0.6204 & 0.4770 \\
\hline
\end{tabular}
\end{table}
{
\subsubsection{Real data analysis}
As a real data-complement to the intractable-model experiments above, we analyze the \texttt{medpar} dataset from the \texttt{COUNT} package (1{,}495 inpatient stays from the 1991 Arizona Medicare MedPAR files, DRG~112; see \cite{hilbe2025count}). Hospital length of stay (\texttt{los}) is well known to exhibit a heavy right tail and occasional extreme values, making a robust DPD-based analysis particularly attractive for this application. We model \texttt{los} with a Poisson regression,
$$
\texttt{los} \sim \texttt{intercept} + \texttt{hmo} + \texttt{white} + \texttt{age80} + \texttt{died} + \texttt{type}.
$$
We compare our DPD-based LLB with SGD sampler against a standard Bayesian Poisson regression fitted via \texttt{MCMCpack::MCMCpoisson} \citep{martin2011mcmcpack}.

Implementation follows the simulation defaults: initialization at the Poisson MLE, Dirichlet-bootstrap weights via exponential draws, robustness parameter $\alpha=0.5$, $m=100$ pseudo-samples per iteration for the intractable term, and $1{,}000$ Dirichlet-bootstrap replicates for posterior draws. The comparator uses \texttt{MCMCpoisson} with \texttt{mcmc}$=10{,}000$, \texttt{burnin}$=10{,}000$, and \texttt{thin}$=10$.

Predictive performance was summarized by the bulk log predictive density (LPD), computed on the central 95\% of responses (excluding the top 5\% of \texttt{los}). On \texttt{medpar}, the bulk LPD for LLB with SGD was $-5182.8458$, whereas standard Bayes achieved $-5295.9202$, an improvement of $+113.0744$ in favor of our method (larger is better). Taken together with the simulation evidence, this real-world application shows that LLB with SGD delivers superior predictive performance in practice for intractable Poisson GLMs while remaining robust to the presence of extreme hospital stays.
}

\section{Concluding remarks}
\label{sec5}
This study introduces a novel stochastic gradient-based methodology for {approximately} sampling from DPD-based posterior distributions, addressing the fundamental challenges in robust Bayesian inference. Our primary contribution lies in developing a computationally efficient approach that overcomes two significant limitations of existing methods: the intractability of integral terms in DPD for general parametric models and computational burden in high-dimensional settings. The theoretical framework we developed combines the LLB with stochastic gradient descent, providing a robust alternative to traditional MCMC methods while maintaining theoretical guarantees. Through extensive simulation studies, we demonstrated that our method produces samples that accurately represent posterior uncertainty, achieving results comparable to those of established methods, such as MH and HMC algorithms.

A key advantage of our approach is its ability to handle models in which the DPD integral term lacks an explicit form, as demonstrated by its application to GLMs. The simulation results with Poisson regression highlight not only the computational efficiency of the method but also its superior statistical accuracy compared to conventional approaches. Notably, the method maintains a stable estimation performance across varying dimensions, consistently achieving lower mean squared error values and improved coverage probabilities for credible intervals. The practical implications of this study are substantial, offering practitioners a computationally tractable approach for robust Bayesian inference that effectively handles outliers and model misspecification in high-dimensional settings. The success of the method with intractable models, particularly in GLMs, suggests its broad applicability across diverse statistical applications.

Several promising directions for future research have emerged from this study. First, extending the methodology to incorporate informative prior distributions would enable shrinkage estimation, potentially improving parameter estimation in high-dimensional settings, where regularization is crucial. Second, adapting the approach to accommodate random effects in generalized linear mixed models would broadens its applicability to hierarchical data structures, which are common in many applied fields. Third, investigating the stability of the posterior predictive distributions using our algorithm merits attention, particularly given the recent findings of \cite{jackjewson2024StabilityGeneralBayesian} regarding the stabilizing effects of DPD-based posteriors on predictive distributions. Additional research directions could include exploring applications to more complex statistical models such as mixture models. In conclusion, this study provides significant advancements in robust Bayesian computation, offering both theoretical guarantees and practical applicability. The ability of the proposed method to handle high-dimensional and intractable models while maintaining computational efficiency is a valuable tool for modern statistical applications.

\backmatter
\subsection*{Acknowledgements}
{The authors would like to thank the handling editor and the two anonymous reviewers for their helpful comments and suggestions.} The authors are grateful to Prof. Kouji Tahata for his valuable comments and for providing the computational infrastructure essential for this research.

\section*{Declarations}
\begin{itemize}
\item Funding:~This work was supported by JSPS KAKENHI Grant Numbers JP20K03756, JP23K13019, JP23K20592, and JP24K23870. 
\item Conflict of interest/Competing interests: ~Not applicable.
\item Ethics approval and consent to participate: ~Not applicable.
\item Consent for publication: ~All authors have read and agreed to the published version of the manuscript.
\item Data availability: ~Not applicable.
\item Materials availability: ~Not applicable.
\item Code availability: ~The R code used to reproduce the experimental results is available in the following GitHub repository: \url{https://github.com/naruki-sonobe/DPD-Bayes-SGD}.
\item Author contribution:~These authors contributed equally to this work.
\end{itemize}

\bibliography{sn-bibliography}

@article{basu1998robust,
  title={Robust and efficient estimation by minimising a density power divergence},
  author={Basu, Ayanendranath and Harris, Ian R and Hjort, Nils L and Jones, MC},
  journal={Biometrika},
  volume={85},
  number={3},
  pages={549--559},
  year={1998},
  publisher={Oxford University Press}
}

@article{bissiri2016general,
  title={A general framework for updating belief distributions},
  author={Bissiri, Pier Giovanni and Holmes, Chris C and Walker, Stephen G},
  journal={Journal of the Royal Statistical Society: Series B (Statistical Methodology)},
  volume={78},
  number={5},
  pages={1103--1130},
  year={2016},
  publisher={Wiley Online Library}
}

@article{eguchi2001robustifing,
  title={Robustifing maximum likelihood estimation by psi-divergence},
  author={Eguchi, Shinto and Kano, Yutaka},
  journal={ISM Research Memorandum},
  volume={802},
  pages={762--763},
  year={2001}
}

@article{ghosh2016robust,
  title={Robust Bayes estimation using the density power divergence},
  author={Ghosh, Abhik and Basu, Ayanendranath},
  journal={Annals of the Institute of Statistical Mathematics},
  volume={68},
  number={2},
  pages={413--437},
  year={2016},
  publisher={Springer}
}

@article{lyddon2019general,
  title={General {Bayesian} updating and the loss-likelihood bootstrap},
  author={Lyddon, Simon P and Holmes, CC and Walker, SG},
  journal={Biometrika},
  volume={106},
  number={2},
  pages={465--478},
  year={2019},
  publisher={Oxford University Press}
}

@article{okuno2024minimizing,
title = {Minimizing Robust Density Power-Based Divergences for General Parametric Density Models},
author = {Okuno, Akifumi},
date = {2024-10-01},
journal = {Annals of the Institute of Statistical Mathematics},
volume = {76},
number = {5},
pages = {851--875},
issn = {1572-9052},
year = {2024},
doi = {10.1007/s10463-024-00906-9}
}

@article{mihoko2002robust,
  title={Robust blind source separation by beta divergence},
  author={Minami, Mihoko and Eguchi, Shinto},
  journal={Neural computation},
  volume={14},
  number={8},
  pages={1859--1886},
  year={2002},
  publisher={MIT Press One Rogers Street, Cambridge, MA 02142-1209, USA journals-info~…}
}

@article{newton1994approximate,
  title={Approximate {Bayesian} inference with the weighted likelihood bootstrap},
  author={Newton, Michael A and Raftery, Adrian E},
  journal={Journal of the Royal Statistical Society Series B: Statistical Methodology},
  volume={56},
  number={1},
  pages={3--26},
  year={1994},
  publisher={Oxford University Press}
}

@article{newton2021weighted,
  title={Weighted {Bayesian} bootstrap for scalable posterior distributions},
  author={Newton, Michael A and Polson, Nicholas G and Xu, Jianeng},
  journal={Canadian Journal of Statistics},
  volume={49},
  number={2},
  pages={421--437},
  year={2021},
  publisher={Wiley Online Library}
}

@article{jewson2018principles,
  title={Principles of {Bayesian} inference using general divergence criteria},
  author={Jewson, Jack and Smith, Jim Q and Holmes, Chris},
  journal={Entropy},
  volume={20},
  number={6},
  pages={442},
  year={2018},
  publisher={Multidisciplinary Digital Publishing Institute}
}

@article{nakagawa2020robust,
  title={Robust {Bayesian} inference via $\gamma$-divergence},
  author={Nakagawa, Tomoyuki and Hashimoto, Shintaro},
  journal={Communications in Statistics-Theory and Methods},
  volume={49},
  number={2},
  pages={343--360},
  year={2020},
  publisher={Taylor \& Francis}
}

@article{matsubara2022robust,
  title = {Robust Generalised {Bayesian} Inference for Intractable Likelihoods},
  author = {Matsubara, Takuo and Knoblauch, Jeremias and Briol, Fran{\c c}ois-Xavier and Oates, Chris J.},
  year = {2022},
  month = apr,
  journal = {Journal of the Royal Statistical Society Series B: Statistical Methodology},
  volume = {84},
  number = {3},
  pages = {997--1022},
  issn = {1369-7412},
  doi = {10.1111/rssb.12500}
}

@article{fujisawa2008robust,
  title = {Robust parameter estimation with a small bias against heavy contamination},
  journal = {Journal of Multivariate Analysis},
  volume = {99},
  number = {9},
  pages = {2053-2081},
  year = {2008},
  issn = {0047-259X},
  doi = {10.1016/j.jmva.2008.02.004},
  author = {Hironori Fujisawa and Shinto Eguchi}
}

@inproceedings{jewson2023beta,
  author = {Jewson, Jack and Ghalebikesabi, Sahra and Holmes, Chris C},
  booktitle = {Advances in Neural Information Processing Systems},
  editor = {A. Oh and T. Naumann and A. Globerson and K. Saenko and M. Hardt and S. Levine},
  pages = {76974--77001},
  publisher = {Curran Associates, Inc.},
  title = {Differentially Private Statistical Inference through $\beta$-Divergence One Posterior Sampling},
  url = {https://proceedings.neurips.cc/paper\_files/paper/2023/file/f3024ea88cec9f45a411cf4d51ab649c-Paper-Conference.pdf},
  volume = {36},
  year = {2023}
}

@inproceedings{knoblauch2018beta,
  author = {Knoblauch, Jeremias and Jewson, Jack E and Damoulas, Theodoros},
  booktitle = {Advances in Neural Information Processing Systems},
  editor = {S. Bengio and H. Wallach and H. Larochelle and K. Grauman and N. Cesa-Bianchi and R. Garnett},
  pages = {},
  publisher = {Curran Associates, Inc.},
  title = {Doubly Robust {Bayesian} Inference for Non-Stationary Streaming Data with $\beta$-Divergences},
  url = {https://proceedings.neurips.cc/paper\_files/paper/2018/file/a3f390d88e4c41f2747bfa2f1b5f87db-Paper.pdf},
  volume = {31},
  year = {2018}
}

@article{warwick2005tuning,
  author = {J. Warwick and M. C. Jones},
  title = {Choosing a robustness tuning parameter},
  journal = {Journal of Statistical Computation and Simulation},
  volume = {75},
  number = {7},
  pages = {581--588},
  year = {2005},
  publisher = {Taylor \& Francis},
  doi = {10.1080/00949650412331299120}
}

@article{basak2021tuning,
  author = {Sancharee Basak and Ayanendranath Basu and M. C. Jones},
  title = {On the ‘optimal’ density power divergence tuning parameter},
  journal = {Journal of Applied Statistics},
  volume = {48},
  number = {3},
  pages = {536--556},
  year = {2021},
  publisher = {Taylor \& Francis},
  doi = {10.1080/02664763.2020.1736524},
  note ={PMID: 35706540}
}

@article{yonekura2023adaptation,
  title={Adaptation of the tuning parameter in general {Bayesian} inference with robust divergence},
  author={Yonekura, Shouto and Sugasawa, Shonosuke},
  journal={Statistics and Computing},
  volume={33},
  number={39},
  year={2023},
  publisher={Springer},
  doi={10.1007/s11222-023-10205-7}
}

@article{ghosh2022robust,
  ISSN = {10170405, 19968507},
  URL = {https://www.jstor.org/stable/27118797},
  author = {Abhik Ghosh and Tuhin Majumder and Ayanendranath Basu},
  journal = {Statistica Sinica},
  number = {2},
  pages = {787--823},
  publisher = {Institute of Statistical Science, Academia Sinica},
  title = {GENERAL ROBUST BAYES PSEUDO-POSTERIORS: EXPONENTIAL CONVERGENCE RESULTS WITH APPLICATIONS},
  urldate = {2024-11-22},
  volume = {32},
  year = {2022}
}

@misc{Rayana:2016,
author = {Shebuti Rayana},
year = {2016},
title = {{ODDS} Library},
url = {https://shebuti.com/outlier-detection-datasets-odds/},
institution = {Stony Brook University, Department of Computer Sciences}
}

@book{niederreiter1992qmc,
  author = {Niederreiter, Harald},
  title = {Random Number Generation and Quasi-Monte Carlo Methods},
  publisher = {Society for Industrial and Applied Mathematics},
  year = {1992},
  doi = {10.1137/1.9781611970081}
}

@article{kawashima2019robust,
  title = {Robust and Sparse Regression in Generalized Linear Model by Stochastic Optimization},
  author = {Kawashima, Takayuki and Fujisawa, Hironori},
  year = {2019},
  month = dec,
  journal = {Japanese Journal of Statistics and Data Science},
  volume = {2},
  number = {2},
  pages = {465--489},
  issn = {2520-8764},
  doi = {10.1007/s42081-019-00049-9}
}

@article{pei-shienwu2023ComparisonLearningRate,
  title = {A Comparison of Learning Rate Selection Methods in Generalized {Bayesian} Inference},
  author = {Wu, Pei-Shien and Martin, Ryan},
  year = {2023},
  month = mar,
  journal = {Bayesian Analysis},
  volume = {18},
  number = {1},
  pages = {105--132},
  doi = {10.1214/21-BA1302}
}

@article{petergrunwald2017InconsistencyBayesianInference,
  title = {Inconsistency of {Bayesian} Inference for Misspecified Linear Models, and a Proposal for Repairing It},
  author = {Gr{\"u}nwald, Peter and van Ommen, Thijs},
  year = {2017},
  month = dec,
  journal = {Bayesian Analysis},
  volume = {12},
  number = {4},
  pages = {1069--1103},
  doi = {10.1214/17-BA1085}
}

@article{syring2019Calibratinggeneralposterior,
  title = {Calibrating General Posterior Credible Regions},
  author = {Syring, Nicholas and Martin, Ryan},
  year = {2019},
  month = jun,
  journal = {Biometrika},
  volume = {106},
  number = {2},
  pages = {479--486},
  issn = {0006-3444},
  doi = {10.1093/biomet/asy054},
  urldate = {2024-12-17}
}

@article{holmes2017Assigningvaluepower,
  title = {Assigning a Value to a Power Likelihood in a General {Bayesian} Model},
  author = {Holmes, C. C. and Walker, S. G.},
  year = {2017},
  month = jun,
  journal = {Biometrika},
  volume = {104},
  number = {2},
  pages = {497--503},
  issn = {0006-3444},
  doi = {10.1093/biomet/asx010},
  urldate = {2024-12-17}
}

@article{walker2016bayes,
  title = {{Bayesian} Information in an Experiment and the {Fisher} Information Distance},
  author = {Walker, Stephen G.},
  year = {2016},
  month = may,
  journal = {Statistics \& Probability Letters},
  volume = {112},
  pages = {5--9},
  issn = {0167-7152},
  doi = {10.1016/j.spl.2016.01.014}
}

@article{ghadimi2013sgd,
  title = {Stochastic First- and Zeroth-Order Methods for Nonconvex Stochastic Programming},
  author = {Ghadimi, Saeed and Lan, Guanghui},
  year = {2013},
  journal = {SIAM Journal on Optimization},
  volume = {23},
  number = {4},
  eprint = {https://doi.org/10.1137/120880811},
  pages = {2341--2368}
}

@article{jackjewson2024StabilityGeneralBayesian,
  title = {On the Stability of General {Bayesian} Inference},
  author = {Jewson, Jack and Smith, Jim Q. and Holmes, Chris},
  year = {2024},
  month = jan,
  journal = {Bayesian Analysis},
  pages = {1--31},
  doi = {10.1214/24-BA1502}
}

@misc{frazier2024loss,
  title={The Impact of Loss Estimation on {Gibbs} Measures}, 
  author={David T. Frazier and Jeremias Knoblauch and Christopher Drovandi},
  year={2024},
  eprint={2404.15649},
  archivePrefix={arXiv},
  primaryClass={math.ST},
  url={https://arxiv.org/abs/2404.15649}
}

@article{dawid2016score,
author = {Dawid, A. Philip and Musio, Monica and Ventura, Laura},
title = {Minimum Scoring Rule Inference},
journal = {Scandinavian Journal of Statistics},
volume = {43},
number = {1},
pages = {123-138},
doi = {https://doi.org/10.1111/sjos.12168},
year = {2016}
}

@InProceedings{fong2019posterior,
  title = 	 {Scalable Nonparametric Sampling from Multimodal Posteriors with the Posterior Bootstrap},
  author =       {Fong, Edwin and Lyddon, Simon and Holmes, Chris},
  booktitle = 	 {Proceedings of the 36th International Conference on Machine Learning},
  pages = 	 {1952--1962},
  year = 	 {2019},
  editor = 	 {Chaudhuri, Kamalika and Salakhutdinov, Ruslan},
  volume = 	 {97},
  series = 	 {Proceedings of Machine Learning Research},
  month = 	 {09--15 Jun},
  publisher =    {PMLR},
  url = 	 {https://proceedings.mlr.press/v97/fong19a.html}
}

@article{giummole2019score,
  title = {Objective {{Bayesian}} Inference with Proper Scoring Rules},
  author = {Giummol{\`e}, F. and Mameli, V. and Ruli, E. and Ventura, L.},
  year = {2019},
  month = sep,
  journal = {TEST},
  volume = {28},
  number = {3},
  pages = {728--755},
  issn = {1863-8260},
  doi = {10.1007/s11749-018-0597-z}
}

@article{pacchiardi2024general,
author = {Lorenzo Pacchiardi and Sherman Khoo and Ritabrata Dutta},
title = {{Generalized Bayesian likelihood-free inference}},
volume = {18},
journal = {Electronic Journal of Statistics},
number = {2},
publisher = {Institute of Mathematical Statistics and Bernoulli Society},
pages = {3628--3686},
year = {2024},
doi = {10.1214/24-EJS2283},
URL = {https://doi.org/10.1214/24-EJS2283}
}

@InProceedings{dellaporta2022robust,
  title = 	 { Robust {Bayesian} Inference for Simulator-based Models via the {MMD} Posterior Bootstrap },
  author =       {Dellaporta, Charita and Knoblauch, Jeremias and Damoulas, Theodoros and Briol, Francois-Xavier},
  booktitle = 	 {Proceedings of The 25th International Conference on Artificial Intelligence and Statistics},
  pages = 	 {943--970},
  year = 	 {2022},
  editor = 	 {Camps-Valls, Gustau and Ruiz, Francisco J. R. and Valera, Isabel},
  volume = 	 {151},
  series = 	 {Proceedings of Machine Learning Research},
  month = 	 {28--30 Mar},
  publisher =    {PMLR},
  url = 	 {https://proceedings.mlr.press/v151/dellaporta22a.html}
}

@article{miller2021general,
  author  = {Jeffrey W. Miller},
  title   = {Asymptotic Normality, Concentration, and Coverage of Generalized Posteriors},
  journal = {Journal of Machine Learning Research},
  year    = {2021},
  volume  = {22},
  number  = {168},
  pages   = {1--53},
  url     = {http://jmlr.org/papers/v22/20-469.html}
}

@inproceedings{lyddon2018nonpara,
author = {Lyddon, Simon and Walker, Stephen and Holmes, Chris},
title = {Nonparametric learning from {Bayesian} models with randomized objective functions},
year = {2018},
publisher = {Curran Associates Inc.},
address = {Red Hook, NY, USA},
booktitle = {Proceedings of the 32nd International Conference on Neural Information Processing Systems},
pages = {2075–2085},
numpages = {11},
location = {Montr\'{e}al, Canada},
series = {NIPS'18}
}

@Manual{hilbe2025count,
  title = {COUNT: Functions, Data and Code for Count Data},
  author = {Joseph M Hilbe},
  year = {2025},
  note = {R package version 1.3.5},
  url = {https://CRAN.R-project.org/package=COUNT},
}

@article{martin2011mcmcpack,
 title={{MCMC}pack: Markov Chain {Monte} {Carlo} in {R}},
 volume={42},
 doi={10.18637/jss.v042.i09},
 number={9},
 journal={Journal of Statistical Software},
 author={Martin, Andrew D. and Quinn, Kevin M. and Park, Jong Hee},
 year={2011},
 pages={1–21}
}

@article{carpenter2017stan,
  title   = {{Stan}: A Probabilistic Programming Language},
  author  = {Carpenter, Bob and Gelman, Andrew and Hoffman, Matthew D. and
             Lee, Daniel and Goodrich, Ben and Betancourt, Michael and
             Brubaker, Marcus and Guo, Jiqiang and Li, Peter and Riddell, Allen},
  journal = {Journal of Statistical Software},
  year    = {2017},
  volume  = {76},
  number  = {1},
  pages   = {1--32},
  doi     = {10.18637/jss.v076.i01}
}

@article{hooker2014robust,
  title = {{Bayesian} Model Robustness via Disparities},
  author = {Hooker, Giles and Vidyashankar, Anand N.},
  year = 2014,
  month = sep,
  journal = {TEST},
  volume = {23},
  number = {3},
  pages = {556--584},
  issn = {1863-8260},
  doi = {10.1007/s11749-014-0360-z}
}
\newpage

\vspace{1cm}
\begin{center}
{\LARGE
{\bf Supplementary Materials for ``Sampling from density power divergence-based posterior distribution via stochastic optimization"}
}
\end{center}

This Supplementary Materials augment the main paper with additional theory and numerical experiments on the DPD-based posterior with LLB–SGD. 
Section~\ref{a1} provides the proofs of the convergence conditions for our method in exponential-families. 
Section~\ref{a2} demonstrates that our method delivers robust density estimation for analytically intractable models (Inverse Gaussian and Gompertz) in the presence of contamination.
Section~\ref{a3} present simulations for parametric (Normal, Poisson) and generalized linear model (GLM) settings, assessing MSE, coverage, and interval length with sensitivity to the sample size, the robustness parameter $\alpha$ and the particle count $m$. 
Section~\ref{a4} shows that DPD–LLB is robust relative to standard Bayes and closely matches MCMC for the same DPD-based posterior.

\begin{appendix}
{
\section{Confirmation of the convergence conditions for our method}\label{a1}
In this section, we prove that the solution generated by our algorithm converges in probability to a minimizer of the objective function $\mathcal{L}_w$ in probability. Let $v>0$ and $\|\cdot\|$ denote the Euclidean norm (its induced operator norm on matrices). We require $\mathcal{L}_w$ to satisfy the following four conditions.
\begin{itemize}
\item[(i)] $\mathcal{L}_w(\theta)$ is differentiable over $\Theta$.
\item[(ii)] There exists $L>0$ such that for any $\theta,\theta'\in\Theta$,
\[
\|\nabla_\theta \mathcal{L}_w(\theta)-\nabla_\theta \mathcal{L}_w(\theta')\|
\le L\|\theta-\theta'\|.
\]
\item[(iii)] $\E_{\xi^{(t)}}\!\big[g(\theta^{(t)}\mid\xi^{(t)})\big]
=\nabla_\theta\mathcal{L}_w(\theta^{(t)})$ for all $t=1,\ldots,T$.
\item[(iv)] $\E_{\xi^{(t)}}\!\big[\|g(\theta^{(t)}\mid\xi^{(t)})-\nabla_\theta\mathcal{L}_w(\theta^{(t)})\|^2\big]\le v$ for all $t=1,\ldots,T$.
\end{itemize}
Below, we work within the exponential-family setup. We first state sufficient assumptions that ensure (a)–(d) in Proposition \ref{prp:abcd} hold. These results (a)–(d) will then be used to establish (ii)–(iv).
Throughout, we take (i) as given and focus on proving (ii)–(iv).

Consider the exponential family
\[
f(x\mid \theta)
= h(x)\exp\!\left\{\ip{\eta(\theta)}{T(x)} - A(\theta)\right\},~
\theta\in\Theta\subset\R^p,
\]
where $T(x)\in\R^k$ is the sufficient statistic, $\eta(\theta)\in\R^k$ is the natural parameter, and $\ip{\cdot}{\cdot}$ denotes the inner product.
The score (the $\theta$-gradient) is
\[
u(x\mid\theta)=\nabla_\theta\log f(x\mid\theta)
= J_\eta(\theta)^\top\!\{T(x)-\mu(\theta)\},~
\mu(\theta)\coloneqq\nabla_\eta A(\eta(\theta))=\E_{F_\theta}[T(X)],
\]
where $J_\eta(\theta)=\partial \eta(\theta)/\partial \theta^\top$.
Fix $\alpha>0$ and write $s(x\mid\theta)\coloneqq f(x\mid\theta)^{\alpha}\,u(x\mid\theta)$.

\begin{asm}\label{ass:base}\mbox{}\\
(i) $\eta(\cdot),A(\cdot)$ are $C^1$.\\
(ii) $\Theta$ is compact and $\{\eta(\theta):\theta\in\Theta\}$ is contained in the interior of the natural parameter space.\\
(iii) Hence $C_J\coloneqq\sup_{\theta\in\Theta}\norm{J_\eta(\theta)}<\infty$ and $C_\mu\coloneqq\sup_{\theta\in\Theta}\norm{\mu(\theta)}<\infty$.
\end{asm}

Let $\Lambda_\gamma\coloneqq\{\gamma\,\eta(\theta):\theta\in\Theta\}$, assumed to be contained in a compact subset of the interior of the natural parameter space. For discrete sample spaces, replace integrals by sums.

\begin{asm}\label{ass:envelope}
\begin{align*}
\sup_{\lambda\in\Lambda_{\alpha}}
\int h(x)^{\alpha} e^{\ip{\lambda}{T(x)}}\{1+\norm{T(x)}\}\,dx &<\infty,\\
\sup_{\lambda\in\Lambda_{2\alpha+1}}
\int h(x)^{2\alpha+1} e^{\ip{\lambda}{T(x)}}\{1+\norm{T(x)}^2\}\,dx &<\infty.
\end{align*}
\end{asm}
These hold for many exponentially-tailed families (Gaussian, Poisson, Gamma, Negative Binomial, Laplace, Inverse Gaussian, etc.) when variances/rates have positive lower bounds and probabilities are bounded away from $0$ and $1$, keeping parameters in the interior. Uniform boundedness of $T(x)$ is not needed; the key is finiteness of weighted moments.

\begin{lem}\label{lem:score}\mbox{}\\
Under Assumption~\ref{ass:base},
\(
\norm{u(x\mid\theta)}\le C_J\{\norm{T(x)}+C_\mu\}.
\)
\end{lem}
\begin{proof}
Immediate from $u(x\mid\theta)=J_\eta(\theta)^\top (T(x)-\mu(\theta))$ and the induced-operator-norm bound.
\end{proof}

\begin{prp}\label{prp:abcd}\mbox{}\\
Under Assumptions~\ref{ass:base}-\ref{ass:envelope}, the following hold (with respect to $\theta$):
\begin{itemize}
\item[(a)] $f(x\mid\theta)^{\alpha}u(x\mid\theta)$ is Lipschitz continuous (for each fixed $x$).
\item[(b)] $f(x\mid\theta)$ is Lipschitz continuous (for each fixed $x$).
\item[(c)] $f(x\mid\theta)^{\alpha}\,\norm{u(x\mid\theta)}$ is integrable (uniformly in $\theta$).
\item[(d)] $f(x\mid\theta)^{2\alpha+1}\,\norm{u(x\mid\theta)}^2$ is integrable (uniformly in $\theta$).
\end{itemize}
\end{prp}

\begin{proof}[Proof of (c) and (d)]
Using Lemma~\ref{lem:score} and $f(x\mid\theta)^\gamma = h(x)^\gamma \exp\{\ip{\gamma\eta(\theta)}{T(x)}-\gamma A(\theta)\}$,
with $C_A\coloneqq\sup_{\theta}e^{-\gamma A(\theta)}<\infty$,
\begin{align*}
\int f(x\mid\theta)^\alpha \norm{u(x\mid\theta)} \,dx &\le C_J\!\left( \int f(x\mid\theta)^\alpha \norm{T(x)} \,dx + C_\mu\int f(x\mid\theta)^\alpha \,dx \right)\\
&\le C_J C_A\!\left( \int h(x)^\alpha e^{\ip{\lambda}{T(x)}}\!\norm{T(x)} \,dx
+ C_\mu \int h(x)^\alpha e^{\ip{\lambda}{T(x)}} \,dx\right),
\end{align*}
where $\lambda=\alpha\eta(\theta)\in\Lambda_\alpha$; finiteness follows from the first line of Assumption~\ref{ass:envelope}.
Similarly,
\begin{align*}
&\int f(x\mid\theta)^{2\alpha+1}\,\|u(x\mid\theta)\|^2 \,dx
\le C_J^2 C_A \Bigg(
    \int h(x)^{2\alpha+1} e^{\langle \lambda, T(x)\rangle} \|T(x)\|^2 \,dx \\
&\qquad + 2C_\mu \int h(x)^{2\alpha+1} e^{\langle \lambda, T(x)\rangle} \|T(x)\| \,dx
    + C_\mu^2 \int h(x)^{2\alpha+1} e^{\langle \lambda, T(x)\rangle} \,dx
\Bigg),
\end{align*}
with $\lambda=(2\alpha+1)\eta(\theta)\in\Lambda_{2\alpha+1}$; finiteness follows from the second line of Assumption~\ref{ass:envelope}.
\end{proof}

\begin{proof}[Proof of (b)]
Since $\nabla_\theta f(x\mid\theta) = f(x\mid\theta)\,u(x\mid\theta)$,
\(
\norm{\nabla_\theta f(x\mid\theta)}
\le \big(\sup_{\vartheta\in\Theta} f(x\mid\vartheta)\big)\, C_J(\norm{T(x)}+C_\mu)<\infty,
\)
where $\sup_\Theta f(x\mid\theta)<\infty$ (for each fixed $x$) follows from compactness of $\Theta$ and continuity.
The mean value theorem yields Lipschitzness.
\end{proof}

\begin{proof}[Proof of (a)]
$\nabla_\theta(f(x\mid\theta)^\alpha u(x\mid\theta))=f(x\mid\theta)^\alpha\{\alpha\,u(x\mid\theta)\,u(x\mid\theta)^\top+\nabla_\theta u(x\mid\theta)\}$ and
$\nabla_\theta u(x\mid\theta) = H_\eta(\theta)^\top(T(x)-\mu(\theta))-J_\eta(\theta)^\top V(\theta) J_\eta(\theta)$ (with $V(\theta)=\Var_{F_\theta}[T(X)]$).
Under Assumption~\ref{ass:base}, $J_\eta,H_\eta,V,\mu$ are bounded; using the first/second-order envelopes coming from (c)(d), which control $f(x\mid\theta)^\alpha, f(x\mid\theta)^\alpha\norm{T(x)}, f(x\mid\theta)^\alpha\norm{T(x)}^2$, we get
$\sup_{\theta\in\Theta}\norm{\nabla_\theta(f(x\mid\theta)^\alpha u(x\mid\theta))}<\infty$ and thus Lipschitzness by the mean value theorem.
\end{proof}

\medskip
\begin{prp}[Lipschitz gradient: (ii)]\label{prp:ii}\mbox{}\\
Let
\[
\nabla_\theta\mathcal{L}_w(\theta)=\sum_{i=1}^n w_i\, f(x_i\mid\theta)^\alpha u(x_i\mid\theta)
\;+\; \E_{F_\theta}\!\left[f(X\mid\theta)^\alpha u(X\mid\theta)\right].
\]
Under Assumptions~\ref{ass:base}-\ref{ass:envelope}, $\nabla_\theta\mathcal{L}_w(\theta)$ is Lipschitz continuous on $\Theta$.
\end{prp}

\begin{proof}
The sample part is Lipschitz by (a). For the expectation part, write
\[
\left\|\E_{F_\theta}[s(X\mid\theta)]-\E_{F_{\theta'}}[s(X\mid\theta')]\right\|
\le \int \|f(x\mid\theta)^{\alpha+1}u(x\mid\theta)-f(x\mid\theta')^{\alpha+1}u(x\mid\theta')\|\,dx,
\]
and then use (a)–(b) pointwise together with the integrability control in (c) to apply the mean value bound under the integral. This yields the stated Lipschitz bound.
\end{proof}

\begin{prp}[Unbiasedness: (iii)]\label{lem:unbiasedness}\mbox{}\\
Let $\xi^{(t)}=\{z_1^{(t)},\dots,z_m^{(t)}\}$ be i.i.d.\ from $F_{\theta^{(t)}}$. Define
\[
g(\theta^{(t)}\mid\xi^{(t)})=
-\sum_{i=1}^n w_i f(x_i\mid\theta^{(t)})^\alpha u(x_i\mid\theta^{(t)})
+\frac{1}{m}\sum_{j=1}^m f(z_j^{(t)}\mid\theta^{(t)})^\alpha u(z_j^{(t)}\mid\theta^{(t)}).
\]
Then $\E_{\xi^{(t)}}[g(\theta^{(t)}\mid\xi^{(t)})]=\nabla_\theta\mathcal{L}_w(\theta^{(t)})$.
\end{prp}

\begin{proof}
By definition,
\[
\nabla_\theta\mathcal{L}_w(\theta)=
-\sum_{i=1}^n w_i f(x_i\mid\theta)^\alpha u(x_i\mid\theta)
\;+\; \E_{F_\theta}\!\left[f(X\mid\theta)^\alpha u(X\mid\theta)\right].
\]
Taking expectation with respect to $\xi^{(t)}$ and using linearity of expectation,
\[
\E_{\xi^{(t)}}\!\left[\frac{1}{m}\sum_{j=1}^m f(z_j^{(t)}\mid\theta^{(t)})^\alpha u(z_j^{(t)}\mid\theta^{(t)})\right]
=\E_{F_{\theta^{(t)}}}\!\left[f(X\mid\theta^{(t)})^\alpha u(X\mid\theta^{(t)})\right],
\]
since $z_j^{(t)}\stackrel{\text{i.i.d.}}{\sim}F_{\theta^{(t)}}$. Combining the fixed (non-random) data term and the expectation of the model term yields
$\E_{\xi^{(t)}}[g(\theta^{(t)}\mid\xi^{(t)})]=\nabla_\theta\mathcal{L}_w(\theta^{(t)})$.
\end{proof}

\begin{prp}[GLM version of (iii)]\label{lem:unbiasedness_glm}\mbox{}\\
In the GLM setting with observation-specific distributions $F_{i,\theta}$, define
\[
g(\theta^{(t)}\mid\xi^{(t)})=
\sum_{i=1}^n w_i\Biggl\{-f(y_i\mid \mathbf{x}_i,\theta^{(t)})^\alpha u(y_i\mid \mathbf{x}_i,\theta^{(t)})
+\frac{1}{m}\sum_{j=1}^m f(z_{ij}\mid \mathbf{x}_i,\theta^{(t)})^\alpha u(z_{ij}\mid \mathbf{x}_i,\theta^{(t)})\Biggr\},
\]
where $z_{ij}\sim F_{i,\theta^{(t)}}$ independently over $i,j$. Then $\E_{\xi^{(t)}}[g(\theta^{(t)}\mid\xi^{(t)})]=\nabla_\theta\mathcal{L}_w(\theta^{(t)})$.
\end{prp}

\begin{proof}
For each $i$, by independence and identical distribution $z_{ij}\sim F_{i,\theta^{(t)}}$,
\[
\E_{\xi^{(t)}}\!\left[\frac{1}{m}\sum_{j=1}^m f(z_{ij}\mid \mathbf{x}_i,\theta^{(t)})^\alpha u(z_{ij}\mid \mathbf{x}_i,\theta^{(t)})\right]
=\E_{F_{i,\theta^{(t)}}}\!\left[f(Y\mid \mathbf{x}_i,\theta^{(t)})^\alpha u(Y\mid \mathbf{x}_i,\theta^{(t)})\right].
\]
Summing over $i$ with weights $w_i$ gives
\[
\E_{\xi^{(t)}}[g(\theta^{(t)}\mid\xi^{(t)})]
=\sum_{i=1}^n w_i\!\left\{-f(y_i\mid \mathbf{x}_i,\theta^{(t)})^\alpha u(y_i\mid \mathbf{x}_i,\theta^{(t)})
+\E_{F_{i,\theta^{(t)}}}[f(Y\mid \mathbf{x}_i,\theta^{(t)})^\alpha u(Y\mid \mathbf{x}_i,\theta^{(t)})]\right\},
\]
which equals $\nabla_\theta\mathcal{L}_w(\theta^{(t)})$ by the GLM definition of $\mathcal{L}_w$.
\end{proof}

\begin{prp}[Second-moment bound: (iv)]\label{prp:iv}\mbox{}\\
For i.i.d.\ $\xi^{(t)}=\{z_1^{(t)},\dots,z_m^{(t)}\}\sim F_{\theta^{(t)}}$,
\[
\mathbb{E}_{\xi^{(t)}}\!\left[
\left\|\frac{1}{m}\sum_{j=1}^m s(z_j^{(t)}\mid\theta^{(t)})
-\E_{F_{\theta^{(t)}}}[s(X\mid\theta^{(t)})]\right\|^2\right]
=\frac{1}{m}\,\mathrm{tr}\big(\Sigma_{\theta^{(t)}}\big)
\le \frac{1}{m}\,\E_{F_{\theta^{(t)}}}\!\left[\norm{s(X\mid\theta^{(t)})}^2\right]<\infty,
\]
where $\Sigma_{\theta}\coloneqq\Var_{F_\theta}(s(X\mid\theta))$.
\end{prp}

\begin{proof}
Let $S_j\coloneqq s(z_j^{(t)}\mid\theta^{(t)})$ and $\bar S_m\coloneqq \frac{1}{m}\sum_{j=1}^m S_j$.
By independence and identical distribution,
\[
\Var(\bar S_m)=\frac{1}{m}\Var(S_1)=\frac{1}{m}\,\Sigma_{\theta^{(t)}}.
\]
Taking the squared norm in expectation gives
\[
\E\big[\|\bar S_m-\E S_1\|^2\big]
=\E\big[(\bar S_m-\E S_1)^\top(\bar S_m-\E S_1)\big]
=\mathrm{tr}\big(\Var(\bar S_m)\big)
=\frac{1}{m}\,\mathrm{tr}(\Sigma_{\theta^{(t)}}).
\]
Finally, $\mathrm{tr}(\Sigma_{\theta})\le \E\|S_1\|^2$ (since $\mathrm{tr}(\Var Z)=\sum_\ell \Var(Z_\ell)\le \sum_\ell \E[Z_\ell^2]=\E\|Z\|^2$), hence
\[
\E\big[\|\bar S_m-\E S_1\|^2\big]\le \frac{1}{m}\,\E\|S_1\|^2
=\frac{1}{m}\,\E_{F_{\theta^{(t)}}}\!\left[\|s(X\mid\theta^{(t)})\|^2\right].
\]
Finiteness follows from Proposition~\ref{prp:abcd}(d), which ensures $\sup_{\theta\in\Theta}\E_{F_\theta}\|s(X\mid\theta)\|^2<\infty$.
\end{proof}
}

\section{DPD-based Posterior Inference for Intractable Models}\label{a2}
This section examines several distributions in which the DPD-based posterior lacks a closed form. We present the density plots for these distributions when applied to contaminated data. These distributions also appeared in the numerical experiments in \cite{okuno2024minimizing}. The R code used for all the analyses in this section is available in the GitHub repository: \url{https://github.com/naruki-sonobe/DPD-Bayes-SGD}.

\paragraph{Inverse Gaussian Distribution.}
First, we consider a scenario in which the true data-generating process is an Inverse Gaussian distribution $IG(1, 3)$ but the observed data contain contamination from a different distribution
\[
x_1, \dots, x_n \overset{iid}{\sim} 0.95\,IG(1,3) + 0.05\,N(10,0.01).
\]
We adopt the Inverse Gaussian distribution $IG(\mu,\lambda)$ to model the data and estimate the parameter vector $\theta=(\mu,\lambda)$. The initial values for our method are set to the maximum likelihood estimator.
\[
\hat{\mu} = \frac{1}{n} \sum_{i=1}^n x_i, \quad \hat{\lambda} = \frac{n}{\sum_{i=1}^n \{x_i^{-1} - \hat{\mu}^{-1}\}}.
\]

\paragraph{Gompertz Distribution.}
Similarly, we examine the case in which the true data-generating process is a Gompertz distribution, $Gompertz(1, 0.1)$; however, the observed data contain the following contamination:
\[
x_1, \dots, x_n \overset{iid}{\sim} 0.95\,Gompertz(1,0.1) + 0.05\,N(10,0.01).
\]
We adopt the Gompertz distribution $Gompertz(\omega,\lambda)$ to model the data and estimate the parameter vector $\theta=(\omega,\lambda)$. In our method, the initial values are set to the maximum likelihood estimator that satisfies
\begin{align*}
&\hat{\lambda}
= -\,\frac{n \hat{\omega}}{
    \,\sum_{j=1}^n \bigl(1 - \exp(\hat{\omega}\,x_j)\bigr)
  },\\
  &\sum_{i=1}^n x_i
\;+\;
\frac{n}{
  \sum_{j=1}^n \bigl(1 - \exp(\hat{\omega}\,x_j)\bigr)
}
\;\sum_{i=1}^n \Bigl\{
  \frac{1 - \exp(\hat{\omega}\,x_i)}{\hat{\omega}}
  \;+\;
  x_i\,\exp\bigl(\hat{\omega}\,x_i\bigr)
\Bigr\}
\;=\;0.
\end{align*}

\begin{figure}[H]
    \centering
    \begin{minipage}{0.45\textwidth}
        \centering
        \includegraphics[width=\textwidth]{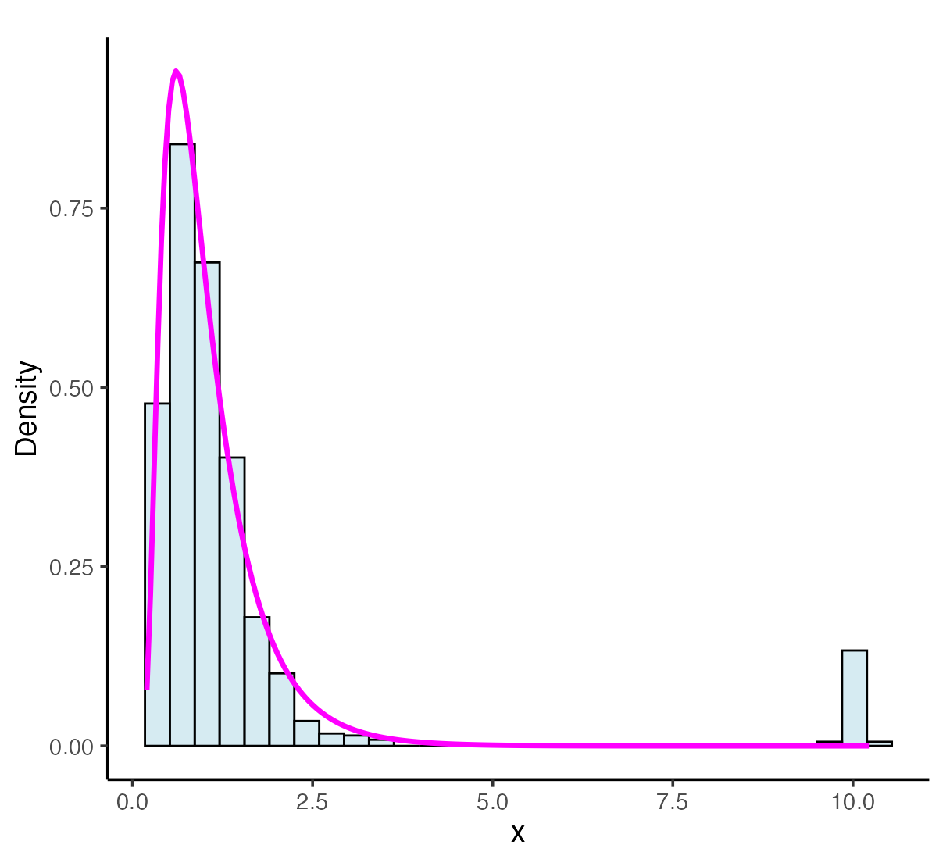} 
    \end{minipage}
    \begin{minipage}{0.45\textwidth}
        \centering
        \includegraphics[width=\textwidth]{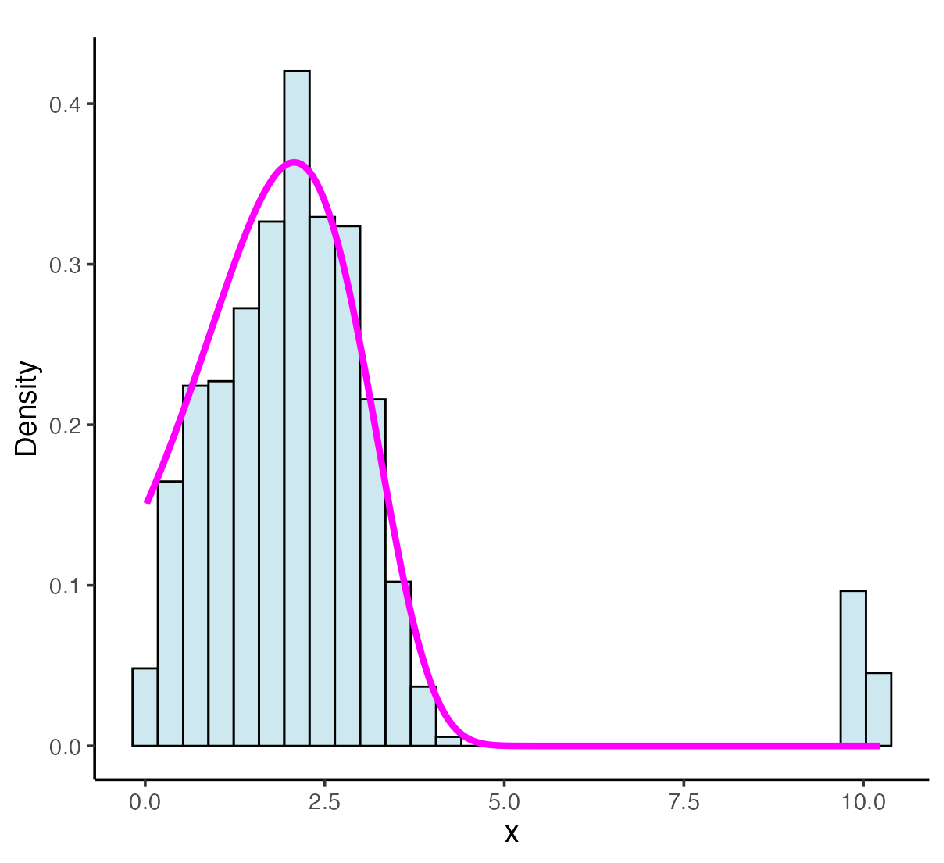}
    \end{minipage}
    \caption{Density estimation for contaminated data. The figures demonstrate that the DPD-based approach (magenta curve) provides robust estimation, effectively mitigating the influence of outliers. Left: Inverse Gaussian distribution. Right: Gompertz distribution.}
    \label{densPlotIntractable}
\end{figure}

The density plots in Figure~\ref{densPlotIntractable} illustrate the robustness of the DPD-based posterior inference in the presence of outliers. For both the Inverse Gaussian and Gompertz distributions, our method appears successfully to recovered the true data-generating distribution despite a 5\% contamination. This demonstrates that the DPD-based approach effectively reduces the influence of outliers while maintaining estimation accuracy for the primary distribution component. Notably, our approach offers significant advantages for models in which the DPD-based posterior cannot be expressed in a closed form. Even in these analytically intractable cases, our computational method enables a robust estimation that remains insensitive to outliers, as is clearly demonstrated by the results of both distributions. 

\section{Sensitivity to $\alpha$ and $m$, and sample size effects}\label{a3}
\subsection{Parametric models}
We conducted extensive simulation studies to evaluate the performance of the proposed DPD-based posterior method across different parametric models under data contamination scenarios. For this analysis, we focused on two common distributional families: the normal model (with an unknown mean \(\mu\) and standard deviation \(\sigma\)) and the Poisson model (with unknown rate parameter \(\lambda\)).

To implement the DPD-based posterior inference, we used robust initial estimates to ensure stability under contamination. For the normal model, we set the initial value of \(\mu\) as the median of the data and the initial value of \(\sigma\) to be as the median absolute deviation (MAD) of the data. For the Poisson model, we used the median of the data as the initial value of \(\lambda\).

For the normal model, we generated data from a contaminated distribution \(0.95 N(0, 1) + 0.05 N(10, 0.01)\), where the majority (95\%) of observations come from a standard normal distribution with true parameters \(\mu=0\) and \(\sigma=1\), whereas a small fraction (5\%) represents outliers from a distribution with mean 10 and very small variance. For the Poisson model, we used a contaminated distribution \(0.95 \text{Po}(1) + 0.05 \text{Po}(10)\), where 95\% of the observations come from a Poisson distribution with a true rate parameter \(\lambda=1\), and 5\% are outliers from a Poisson distribution with a rate parameter of 10. Our goal was to examine how well the DPD-based posterior method could recover the true underlying parameters (\(\mu=0\), \(\sigma=1\) for normal distribution, and \(\lambda=1\) for Poisson distribution) in the presence of these contaminations.

\begin{table}[h]
\centering
\caption{Performance of the proposed DPD-based posterior method for Normal and Poisson models under different sample sizes over 1,000 simulations. All experiments used \(m=10\) particles.}
\label{tab:example1}
\begin{tabular}{|r|rrr|rrr|rrr|}
\hline
 & \multicolumn{6}{c|}{Normal} & \multicolumn{3}{c|}{Poisson} \\
\cmidrule{2-10}
 & \multicolumn{3}{c|}{\(\mu\)} & \multicolumn{3}{c|}{\(\sigma\)} & \multicolumn{3}{c|}{\(\lambda\)} \\
\cmidrule{2-10}
\(n\) & MSE & CP & AL & MSE & CP & AL & MSE & CP & AL \\
\hline
50 & 0.0238 & 0.9570 & 0.6269 & 0.0143 & 0.9440 & 0.4848 & 0.0252 & 0.9510 & 0.6210\\
100 & 0.0127 & 0.9540 & 0.4459 & 0.0068 & 0.9580 & 0.3455 & 0.0134 & 0.9520 & 0.4487 \\
500 & 0.0025 & 0.9650 & 0.2133 & 0.0017 & 0.9570 & 0.1652 & 0.0029 & 0.9550  & 0.2160 \\
1,000 & 0.0012 & 0.9850 & 0.1628 & 0.0011 & 0.9490 & 0.1259 & 0.0017 & 0.9650 & 0.1649 \\
\hline
\end{tabular}
\end{table}
\begin{table}[h]
\centering
\caption{Performance of the proposed DPD-based posterior method for Normal and Poisson models under different numbers of particles (\(m\)) over 1,000 simulations. All experiments used a fixed sample size of \(n=500\).}
\label{tab:example2}
\begin{tabular}{|r|rrr|rrr|rrr|}
\hline
 & \multicolumn{6}{c|}{Normal} & \multicolumn{3}{c|}{Poisson} \\
\cmidrule{2-10}
 & \multicolumn{3}{c|}{\(\mu\)} & \multicolumn{3}{c|}{\(\sigma\)} & \multicolumn{3}{c|}{\(\lambda\)} \\
\cmidrule{2-10}
\(m\) & MSE & CP & AL & MSE & CP & AL & MSE & CP & AL \\
\hline
5 & 0.0025 & 0.9790 & 0.2310 & 0.0018 & 0.9700 & 0.1785 & 0.0030 & 0.9640 & 0.2334 \\
10 & 0.0025 & 0.9670 & 0.2132 & 0.0018 & 0.9490 & 0.1644 & 0.0029 & 0.9530 & 0.2161 \\
50 & 0.0024 & 0.9520 & 0.1987 & 0.0017 & 0.9390 & 0.1533 & 0.0029 & 0.9330  & 0.2013 \\
100 & 0.0024 & 0.9570 & 0.1972 & 0.0017 & 0.9300 & 0.1520 & 0.0032 & 0.9210  & 0.2001 \\
\hline
\end{tabular}
\end{table}
\begin{table}[h]
\centering
\caption{Performance of the proposed DPD-based posterior method for Normal and Poisson models under different robustness parameter (\(\alpha\)) over 1,000 simulations. All experiments used a fixed sample size of \(n=100\) and $m=10$ particles.}
\label{tab:example3}
\begin{tabular}{|r|rrr|rrr|rrr|}
\hline
 & \multicolumn{6}{c|}{Normal} & \multicolumn{3}{c|}{Poisson} \\
\cmidrule{2-10}
 & \multicolumn{3}{c|}{\(\mu\)} & \multicolumn{3}{c|}{\(\sigma\)} & \multicolumn{3}{c|}{\(\lambda\)} \\
\cmidrule{2-10}
\(\alpha\) & MSE & CP & AL & MSE & CP & AL & MSE & CP & AL \\
\hline
0   & 0.2602 & 0.1750 & 0.9383 & \(1.3\times10^{258}\) & 0.0000 & 1.5465 & 0.2198 & 0.2150 & 0.8884 \\
0.1 & 0.0129 & 0.9670 & 0.7082 & 0.0318 & 0.9610 & 1.6370 & 0.0300 & 0.8460 & 0.5111 \\
0.2 & 0.0106 & 0.9510 & 0.4159 & 0.0059 & 0.9510 & 0.3026 & 0.0149 & 0.9310 & 0.4478 \\
0.5 & 0.0129 & 0.9450 & 0.4450 & 0.0072 & 0.9470 & 0.3449 & 0.0128 & 0.9460 & 0.4487 \\
0.7 & 0.0146 & 0.9490 & 0.4680 & 0.0081 & 0.9590 & 0.3770 & 0.0141 & 0.9540 & 0.4630 \\
0.8 & 0.0131 & 0.9650 & 0.4808 & 0.0092 & 0.9490 & 0.3904 & 0.0150 & 0.9360 & 0.4691 \\
\hline
\end{tabular}
\end{table}

Table~\ref{tab:example1} summarizes the results of both models for over 1,000 independent simulation experiments with varying sample sizes. We report three key performance metrics: the mean squared error (MSE) of the posterior mean estimate, coverage probability (CP) of the 95\% credible interval, and average length (AL) of that interval. 

The results demonstrate that the proposed method performs well for different sample sizes. As the sample size increases from \(n=50\) to \(n=1,000\), both MSE and AL consistently improve (decrease) for all parameters. However, we observe that the coverage probability tends to exceed 95\% for larger sample sizes, particularly for the mean parameter of the normal model (\(\mu\)), which reaches 98.5\% coverage at \(n=1,000\). This suggests that the credible intervals produced by the proposed method are conservative in practice, particularly for larger datasets.

We also investigated the effect of varying the number of particles (\(m\)) used in the stochastic gradient algorithm. The results are presented in Table~\ref{tab:example2}. For these experiments, we fixed the sample size at \(n=500\) and varied \(m\) from 5 to 100. Interestingly, the MSE remains relatively stable across different values of \(m\) for all parameters, suggesting that even a small number of particles can provide reasonable point estimates.

However, we observe that the coverage probability tends to decrease as \(m\) increases, particularly for the Poisson model, where CP decreases from 96.4\% with \(m=5\) to 92.1\% with \(m=100\). Similarly, the average length of the credible intervals modestly decreases with an increase \(m\). These patterns suggest that while larger values of \(m\) may lead to more computationally intensive estimation, they also produce narrower credible intervals that are closer to the nominal 95\% coverage.

{
Table~\ref{tab:example3} reports the sensitivity to \(\alpha\). At \(\alpha=0\) (pure likelihood), outliers strongly affect the Normal model: coverage for \(\mu\) drops to 17.5\%, and the variance parameter shows severe inflation. This inflation is also partly due to numerical instability in the optimization at \(\alpha=0\). Making \(\alpha\) slightly positive (\(0.1\)–\(0.2\)) quickly stabilizes behavior: MSE decreases, coverage moves back toward 95\%, and average lengths shorten moderately. For larger \(\alpha\) (\(\ge 0.5\)), coverage remains around 94-96\% with slightly tighter intervals and a small efficiency trade-off. For the Poisson model, accuracy at \(\alpha=0.1\) is relatively poor (e.g., MSE \(=0.030\), CP \(=84.6\%\)); thus, \(\alpha=0.2\) or \(\alpha=0.5\) is recommended in practice.

When $\alpha>0$, the DPD-based posterior provides reliable estimation and uncertainty quantification under the considered contamination: MSE and AL decrease with $n$, CP stays near nominal, and modest particle counts (e.g., $m=5$–$10$) suffice, whereas $\alpha=0$ (the pure likelihood case) remains sensitive to outliers.
}

\subsection{GLMs}
Next, we illustrate the DPD-based posterior inference for GLMs. Extending our analysis beyond simple parametric models, we investigated the performance of our proposed approach in more complex regression settings that are widely used in practice.

Specifically, we focus on two canonical models within the exponential family framework of GLMs:
\begin{itemize}
    \item A linear regression model with normal errors, where we estimate both the regression coefficient vector \(\beta\) and the error variance \(\sigma^2\)
    \item A Poisson regression model for count data, where we estimate the regression coefficient vector \(\beta\) that relates covariates to the mean of the Poisson response
\end{itemize}

To implement our DPD-based posterior inference for both models, we used the maximum likelihood estimates from standard methods (ordinary least squares for linear regression and standard Poisson regression fitting) as the initial parameter values. 

For both models, we generated synthetic datasets with two covariates (\(x_1\) and \(x_2\)) drawn from a standard normal distribution. For the linear regression model, we used a contaminated distribution, where 95\% of the data follow \(N(0.6+0.3x_1+0.2x_2, 1)\) and 5\% follow \(N(10.6+0.3x_1+0.2x_2, 1)\), with true parameter values \(\beta_0 = 0.6\), \(\beta_1 = 0.3\), \(\beta_2 = 0.2\), and \(\sigma = 1\). For the Poisson regression model, we similarly used a contaminated distribution, where 95\% of the data follow \(\mathrm{Po}\bigl(\exp(0.6+0.3x_1+0.2x_2)\bigr)\) and 5\% follow \(\mathrm{Po}\bigl(10 + \exp(0.6+0.3x_1+0.2x_2)\bigr)\), with the same true regression coefficients.

We then applied our DPD-based posterior inference method and evaluated its performance across various sample sizes and particle counts for stochastic gradient estimation. Our primary goal was to assess how well our method could recover the true underlying parameters despite data contamination.

Table~\ref{tab:example_regression1} presents the results for both regression models across different sample sizes, while maintaining a fixed number of particles \(m=10\). We conducted 100 independent simulation replications for each configuration and reported the MSE of the posterior mean estimate, CP of the 95\% credible interval, and AL of that interval.

The results indicate that as the sample size increased from \(n=50\) to \(n=1,000\), the MSE consistently decreased for all parameters in both models, indicating improved estimation accuracy with larger samples. Similarly, AL of the credible intervals decreased with increasing sample size, reflecting greater precision in the posterior distributions. For the regression coefficient \(\beta\), CP remained relatively stable and close to the nominal 95\% level across different sample sizes in both models.

We further investigated the effect of varying the number of particles (\(m\)) used in the stochastic gradient algorithm, and the results are presented in Table~\ref{tab:example_regression2}. For these experiments, we fixed the sample size at \(n=100\) and varied \(m\) from 5 to 100.

\begin{table}[h]
\centering
\caption{Performance of the proposed DPD-based approach for linear and Poisson regressions under different sample sizes over 100 simulations. All experiments used $m=10$ particles.}
\label{tab:example_regression1}
\begin{tabular}{|r|rrr|rrr|rrr|}
\hline
 & \multicolumn{6}{c|}{Linear Regression} & \multicolumn{3}{c|}{Poisson Regression} \\
\cmidrule{2-10}
 & \multicolumn{3}{c|}{$\beta$} & \multicolumn{3}{c|}{$\sigma$} & \multicolumn{3}{c|}{$\beta$} \\
\cmidrule{2-10}
$n$ & MSE & CP & AL & MSE & CP & AL & MSE & CP & AL \\
\hline
50   & 0.0265  & 0.9433 & 0.6606 & 0.0223  & 0.9000 & 0.5337 & 0.0123  & 0.9367 & 0.4694 \\
100  & 0.0141  & 0.9433 & 0.4260 & 0.0091 & 0.9100 & 0.3477 & 0.0074 & 0.9600 & 0.3490 \\
500  & 0.0025 & 0.9367 & 0.1911 & 0.0012 & 0.9800 & 0.1505 & 0.0015 & 0.9500 & 0.1488 \\
1,000 & 0.0013 & 0.9567 & 0.1393 & 0.0010 & 0.9300 & 0.1073 & 0.0007& 0.9433 & 0.1023 \\
\hline
\end{tabular}
\end{table}

\begin{table}[h]
\centering
\caption{Performance of the proposed DPD-based approach for linear and Poisson regressions under different numbers of particles ($m$) over 100 simulations. All experiments used a fixed sample size of $n=100$.}
\label{tab:example_regression2}
\begin{tabular}{|r|rrr|rrr|rrr|}
\hline
 & \multicolumn{6}{c|}{Linear Regression} & \multicolumn{3}{c|}{Poisson Regression} \\
\cmidrule{2-10}
 & \multicolumn{3}{c|}{$\beta$} & \multicolumn{3}{c|}{$\sigma$} & \multicolumn{3}{c|}{$\beta$} \\
\cmidrule{2-10}
$m$ & MSE & CP & AL & MSE & CP & AL & MSE & CP & AL \\
\hline
5    & 0.0108  & 0.9667 & 0.4344 & 0.0077 & 0.9600 & 0.3564 & 0.0101  & 0.9533 & 0.3963 \\
10   & 0.0143  & 0.9567 & 0.4690 & 0.0077 & 0.9400 & 0.3562 & 0.0081 & 0.9567 & 0.3572 \\
50   & 0.0137  & 0.9400 & 0.4555 & 0.0078 & 0.9400 & 0.3472 & 0.0092 & 0.9400 & 0.3808 \\
100  & 0.0126  & 0.9700 & 0.4565 & 0.0063 & 0.9600 & 0.3485 & 0.0069 & 0.9333 & 0.3264 \\
\hline
\end{tabular}
\end{table}

\begin{table}[h]

\centering
\caption{Performance of the proposed DPD-based approach for linear and Poisson regressions under different robustness parameter ($\alpha$) over 100 simulations. All experiments used a fixed sample size of $n=100$ and $m=10$ particles.}
\label{tab:example_regression3}
\begin{tabular}{|r|rrr|rrr|rrr|}
\hline
 & \multicolumn{6}{c|}{Linear Regression} & \multicolumn{3}{c|}{Poisson Regression} \\
\cmidrule{2-10}
 & \multicolumn{3}{c|}{\(\beta\)} & \multicolumn{3}{c|}{\(\sigma\)} & \multicolumn{3}{c|}{\(\beta\)} \\
\cmidrule{2-10}
\(\alpha\) & MSE & CP & AL & MSE & CP & AL & MSE & CP & AL \\
\hline
0   & 0.1967  & 0.5233 & 0.9598 & 1.6672  & 0.0000 & 1.4721 & 0.0477 & 0.4667 & 0.3982 \\
0.1 & 0.0130 & 0.9567 & 0.5818 & 0.0234 & 1.0000 & 1.6892 & 0.0085 & 0.8933 & 0.3304 \\
0.2 & 0.0116 & 0.9433 & 0.4148 & 0.0051 & 0.9700 & 0.2952 & 0.0069 & 0.9433 & 0.3143 \\
0.5 & 0.0146 & 0.9567 & 0.4644 & 0.0072 & 0.9700 & 0.3563 & 0.0059 & 0.9300 & 0.2840 \\
0.7 & 0.0141 & 0.9433 & 0.4882 & 0.0095 & 0.9700 & 0.3924 & 0.0071 & 0.9167 & 0.2775 \\
0.8 & 0.0165 & 0.9500 & 0.4972 & 0.0122 & 0.9100 & 0.4345 & 0.0079 & 0.8933 & 0.2673 \\
\hline
\end{tabular}
\end{table}

Interestingly, the MSE and AL of the credible intervals remained relatively stable across the different values of \(m\) for both regression models. This suggests that our DPD-based inference method is robust to the choice of the particle count in the stochastic gradient estimation, with even a small number of particles (\(m=5\) or \(m=10\)) providing reliable point and interval estimates. The CPs show some fluctuation but generally remain close to the nominal 95\% level across different values of \(m\) for all parameters.

{Table~\ref{tab:example_regression3} further examines sensitivity to the robustness parameter \(\alpha\). At \(\alpha=0\) (pure likelihood), contamination markedly degrades performance: linear-regression variance is unstable and Poisson coverage is poor. Setting \(\alpha>0\) improves performance: both MSE and AL decrease, and CP moves toward the nominal level. For linear regression, the tested settings \(\alpha=0.2\) and \(\alpha=0.5\) provide a good balance between efficiency and robustness. Similarly, for Poisson regression, \(\alpha=0.2\) or \(\alpha=0.5\) achieves near-nominal coverage with small MSE. Overall, based on our evaluated grid, \(\alpha=0.2\) or \(\alpha=0.5\) are practical default choices for robust and efficient inference in these GLMs.

When $\alpha>0$, the DPD-based posterior extends well to the GLM settings studied: MSE and AL improve with larger $n$, CP remains close to nominal, and modest $m$ already performs well, while the pure likelihood case ($\alpha=0$) is vulnerable to contamination.
}

{
\section{Empirical robustness and posterior geometry: DPD–LLB versus baseline methods}\label{a4}
We present two complementary experiments that address robustness and uncertainty quantification from different angles. The first focuses on a univariate Normal location-scale model and contrasts our DPD-LLB posterior with the standard Bayesian posterior under $\varepsilon$-contamination. The second examines a bivariate mean model with known covariance and compares LLB-based sampling with MCMC sampling of the same DPD-based posterior, highlighting how the two sampling strategies agree in clean data and under contamination, including a direct visualisation of credible regions.

\medskip
\noindent\textbf{Experiment A: Univariate contamination and comparison with Standard Bayes.}
For each replication we generated $n=1,000$ observations from
\[
(1-\varepsilon)\,\mathcal N(0,1)\;+\;\varepsilon\,\mathcal N(10,0.1^2),\qquad \varepsilon=0.05,
\]
and fitted the working model $\mathcal N(\mu,\sigma)$ using two methods. For DPD-LLB we set $\alpha=0.5$ and used a stochastic-optimisation scheme with $m=100$ inner draws per iteration, $500$ iterations with scheduled step-size decay, and robust initialisation (median/MAD). For the standard Bayesian posterior we used the Gaussian likelihood with $p(\mu,\sigma)\propto 1/\sigma$ and a random-walk Metropolis-Hastings sampler (proposal s.d.\ $0.05$ on both coordinates, thinning by $50$ to retain $1,000$ draws). The experiment was repeated $20$ times with fresh data.

Kernel-density overlays of the per-replication posteriors (Figure~\ref{fig:normal_dpd_vs_bayes}) show a clear qualitative contrast. Across all replications, the DPD-LLB posterior concentrates near the uncontaminated targets ($\mu\approx0$, $\sigma\approx1$) with minimal between-replication variation; curves from different runs nearly coincide. In contrast, the standard Bayesian posteriors are systematically distorted by the $5\%$ outliers: the posterior for $\mu$ shifts positively and that for $\sigma$ inflates, with noticeable between-replication dispersion that reflects the random number of extreme draws. These patterns indicate that DPD-LLB anchors inference to the uncontaminated bulk of the data, whereas the pure likelihood is pulled toward the contaminating component.

\begin{figure}[H]
    \centering
    \includegraphics[width=0.8\textwidth]{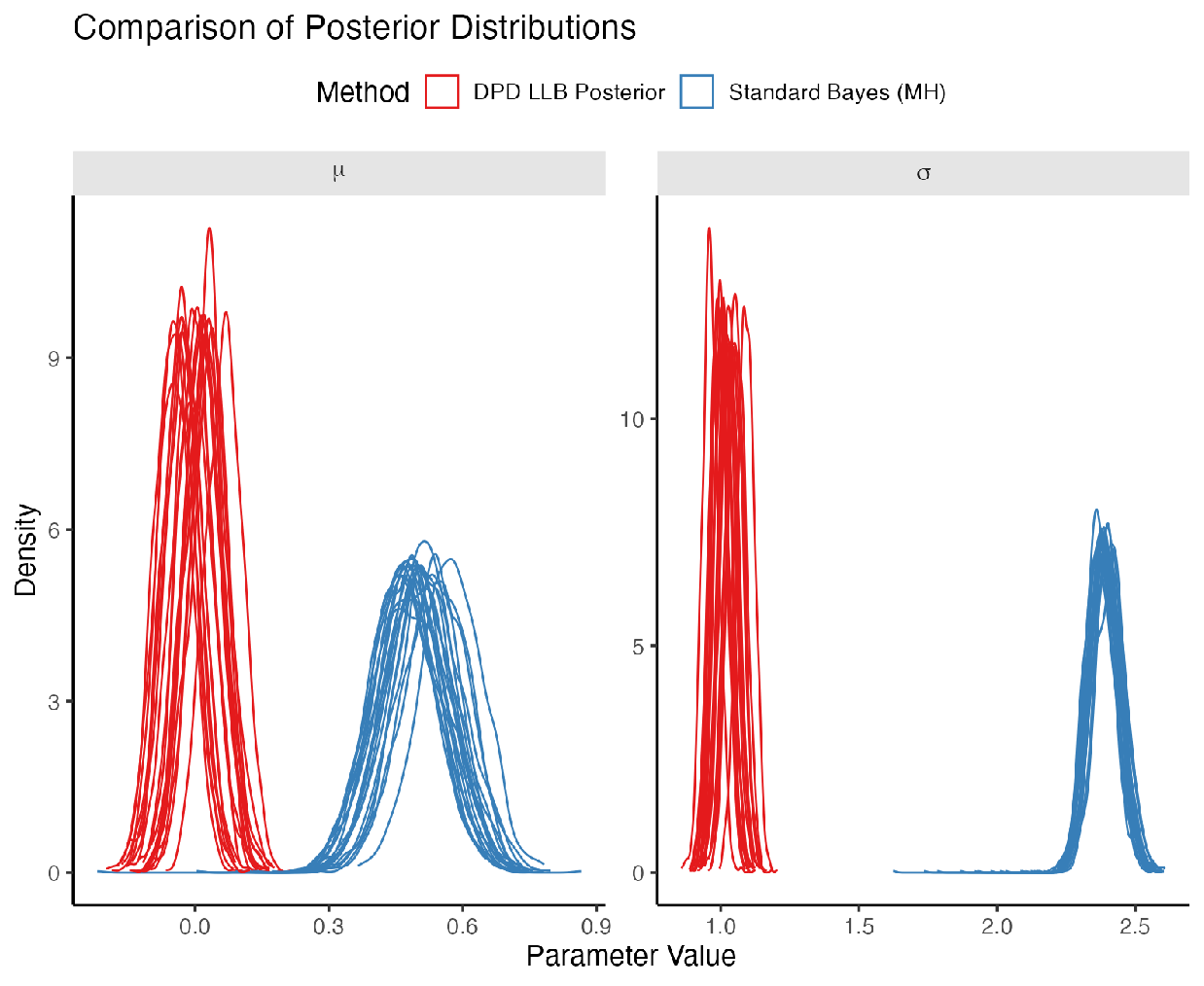}
    \vspace{-2.5mm}
    \caption{Posterior density overlays across $20$ replications for the contaminated Normal model ($n=1,000$, $\varepsilon=0.05$). DPD-LLB ($\alpha=0.5$, $m=100$; red) remains centred near $\mu\approx0$ and $\sigma\approx1$ with little between-replication variation, while the standard Bayesian posterior (MH; blue) is shifted for $\mu$ and inflated for $\sigma$.}
    \label{fig:normal_dpd_vs_bayes}
\end{figure}

\medskip
\noindent\textbf{Experiment B: Two-dimensional credible regions and DPD-based posterior sampling.}
Motivated by the reviewer’s suggestion to enhance empirical analysis, we examined whether LLB-based sampling and MCMC sampling yield similar DPD-based posteriors and credible regions in a two-dimensional setting. We considered a Normal mean model with known covariance $\Sigma=I_2$ and target $\mu=(\mu_1,\mu_2)$, with DPD tuning fixed at $\alpha=0.5$. For LLB (``LLB with SGD'') we used $m=20$ inner draws per iteration, $400$ iterations with scheduled step-size decay, robust median initialisation, and drew $3{,}000$ samples. For the MCMC baseline we sampled the same DPD-based posterior using a Gaussian random-walk MH (step size $0.12$), a burn-in of $1{,}500$, thinning by $50$, and retained $3{,}000$ draws; the DPD weight $w$ was calibrated from numerical $I$– and $J$–matrices around a short LLB pilot. We examined two regimes with $n=800$: a clean sample ($0\%$ outliers) and an $\varepsilon$-contaminated sample ($5\%$ from $N((6,6),0.1^2 I_2)$ with the remaining $95\%$ from $N((0,0),I_2)$).

Figure~\ref{fig:llb_mh_2d} displays the $50\%$ and $95\%$ credible ellipses together with posterior means and the true centre. In the clean case, LLB and MH produce nearly indistinguishable centres and very similar ellipses, indicating that LLB-based sampling accurately reproduces the DPD-based posterior geometry delivered by MCMC when the model is well specified. Under $5\%$ contamination, both procedures remain centred close to $(0,0)$ and their credible regions expand sensibly without drifting toward the outliers; shapes and orientations of the ellipses are again similar. Small scale differences are consistent with the stochastic-gradient approximation and the data-driven calibration of $w$.

\begin{figure}[H]
  \centering
  \includegraphics[width=\linewidth]{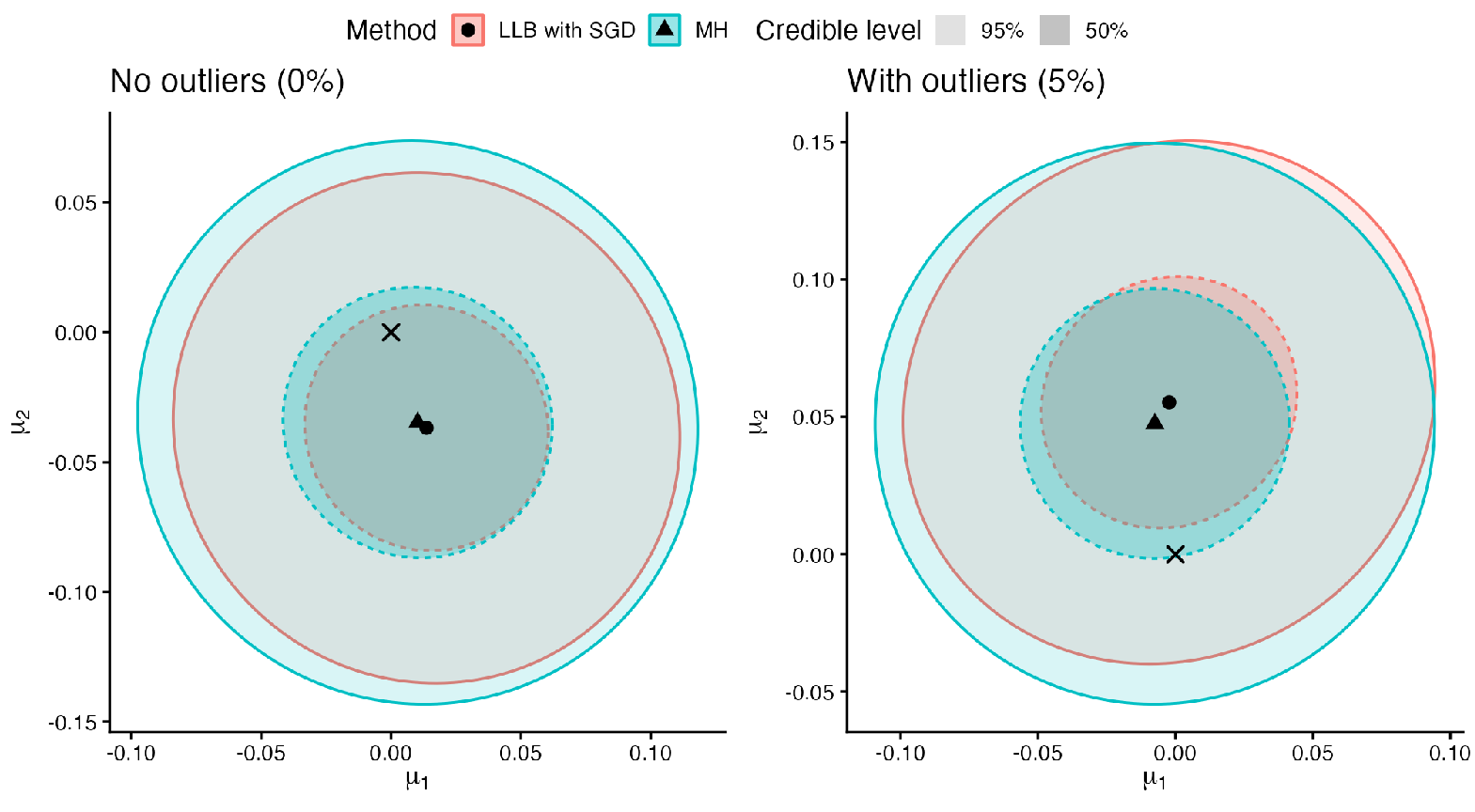}
  \caption{Two-dimensional DPD-based posterior for the mean under clean data (left) and $5\%$ contamination (right) with $n=800$ and $\alpha=0.5$. Shaded regions show $50\%$ and $95\%$ credible ellipses for LLB with SGD and MH; markers indicate posterior means and the true centre $(0,0)$.}
  \label{fig:llb_mh_2d}
\end{figure}

\medskip
\noindent\textbf{Joint conclusions.}
Taken together, the univariate and bivariate results lead to three empirical conclusions. First, DPD-LLB yields robust inference under moderate contamination, keeping posteriors well centred with controlled spread, whereas the standard likelihood-based posterior is materially distorted. Second, LLB-based sampling and MCMC sampling of the same DPD-based posterior agree closely in centre, shape, and orientation, both in clean data and with outliers, so the stochastic-gradient construction preserves the posterior geometry of interest. Third, any residual differences between LLB and MH are small relative to overall posterior uncertainty and do not affect qualitative conclusions about robustness or the calibration of credible regions.
}
\end{appendix}

\end{document}